\documentclass[a4paper,UKenglish,cleveref, autoref, thm-restate, runningheads]{lipics-v2021}

\usepackage[T1]{fontenc}
\usepackage{graphicx}

\usepackage{amssymb}
\usepackage{mathtools}
\usepackage{amsmath}
\usepackage{amsfonts}
\usepackage{mathabx}
\renewcommand{\leq}{\varleq}
\renewcommand{\geq}{\vargeq}
\renewcommand{\epsilon}{\varepsilon}
\usepackage{multirow}
\usepackage{tabularx}
\usepackage{diagbox}
\usepackage{array}

\usepackage{mathcommand}

\usepackage[]{todonotes}
\usepackage{tikz}
\usetikzlibrary{positioning}
\usetikzlibrary{calc}

\usepackage{mdframed}
\usepackage[linesnumbered,vlined,noend,boxed,ruled]{algorithm2e}
\SetKwInput{KwData}{Input}
\SetKwInput{KwResult}{Output}

\usepackage{thm-restate}

\renewcommand{\O}{\mathcal{O}}

\NewDocumentTextCommand\testing{O{$\cdot$}O{$\cdot$}}{Test[#1][#2]}
\NewDocumentMathCommand\testing{O{\cdot}O{\cdot}}{\textup{Test}[#1][#2]}

\newcommand{\abs}[1]{\lvert #1 \rvert}
\usepackage{xparse}

\NewDocumentTextCommand\LS{O{$\cdot$}O{$\cdot$}}{LastSeen[#2][#1]} 
\NewDocumentMathCommand\LS{O{\cdot}O{\cdot}}{\textup{LastSeen}[#2][#1]} 

\NewDocumentTextCommand\KMalt{O{$\cdot$}O{$\cdot$}}{Knotenmatrix[#2][#1]} 
\NewDocumentMathCommand\KMalt{O{\cdot}O{\cdot}}{\textup{Knotenmatrix}[#2][#1]}

\NewDocumentTextCommand\FIA{O{$\cdot$}O{$\cdot$}}{FirstInArch[#1][#2]}
\NewDocumentMathCommand\FIA{O{\cdot}O{\cdot}}{\textup{FirstInArch}[#1][#2]}

\NewDocumentTextCommand\LIA{O{$\cdot$}O{$\cdot$}}{LastInArch[#1][#2]}
\NewDocumentMathCommand\LIA{O{\cdot}O{\cdot}}{\textup{LastInArch}[#1][#2]}

\NewDocumentTextCommand\mA{O{$\cdot$}}{minArch[#1]}
\NewDocumentMathCommand\mA{O{\cdot}}{\textup{minArch}[#1]}

\newtextcommand{\last}{last}
\newmathcommand{\last}{\textup{last}}

\newtextcommand{\startSAS}{startSAS}
\newmathcommand{\startSAS}{\textup{startSAS}}

\newtextcommand{\arches}{$\textup{ar}_w(1)\dots\textup{ar}_w(k)$}
\newmathcommand{\arches}{\textup{ar}_w(1)\dots\textup{ar}_w(k)}

\NewDocumentTextCommand{\sL}{O{$\cdot$}O{$\cdot$}}{sortedLast[#1][#2]}
\NewDocumentMathCommand{\sL}{O{\cdot}O{\cdot}}{\textup{sortedLast}[#1][#2]}

\NewDocumentTextCommand{\sF}{O{$\cdot$}O{$\cdot$}}{sortedFirst[#1][#2]}
\NewDocumentMathCommand{\sF}{O{\cdot}O{\cdot}}{\textup{sortedFirst}[#1][#2]}

\NewDocumentTextCommand{\Leqq}{O{$\cdot$}O{$\cdot$}}{Leq[#1][#2]}
\NewDocumentMathCommand{\Leqq}{O{\cdot}O{\cdot}}{\textup{Leq}[#1][#2]}

\NewDocumentTextCommand{\archIdx}{O{$\cdot$}}{archIdx[#1]}
\NewDocumentMathCommand{\archIdx}{O{\cdot}}{\textup{archIdx}[#1]}

\NewDocumentTextCommand{\knn}{O{$\cdot$}}{knnext[#1]}
\NewDocumentMathCommand{\knn}{O{\cdot}}{\textup{knnext}[#1]}

\NewDocumentTextCommand{\knp}{O{$\cdot$}}{knprev[#1]}
\NewDocumentMathCommand{\knp}{O{\cdot}}{\textup{knprev}[#1]}

\NewDocumentTextCommand{\testt}{d<> O{$\cdot$}}{\IfNoValueTF{#1}{test nr. 1 mit #1 und #2 und so}{zweiter test #1 und #2 und alles}}

\NewDocumentTextCommand{\KM}{d<> o}{\IfNoValueTF{#1}{\IfNoValueTF{#2}{P[$\cdot$, $\cdot$]}{P[#2]}}{P[#2, #1]}}
\NewDocumentMathCommand{\KM}{d<> o}{\IfNoValueTF{#1}{\IfNoValueTF{#2}{{P}[\cdot, \cdot]}{{P}[#2]}}{{P}[#2, #1]}}

\DeclareMathOperator{\al}{alph}
\DeclareMathOperator{\ar}{ar}
\DeclareMathOperator{\dist}{dist}
\DeclareMathOperator{\firstPosArch}{firstInArch}

\DeclareMathOperator{\lastPosArch}{lastInArch}

\DeclareMathOperator{\mas}{MAS}

\DeclareMathOperator{\nextArray}{sameNext}
\DeclareMathOperator{\prevArray}{samePrev}

\DeclareMathOperator{\nextpos}{next}

\DeclareMathOperator{\back}{back}

\DeclareMathOperator{\rest}{r}
\DeclareMathOperator{\sas}{SAS}

\DeclareMathOperator{\defPath}{defPath}
\DeclareMathOperator{\lastGap}{lastGap}
\DeclareMathOperator{\lastPair}{lastPair}

\def\nth#1{#1$^{\text{th}}$}
\newcommand{\len}[1]{|#1|}

\newcommand{\rightedge}[3]{#1 \rightarrow_{#3} #2}

\newcommand{\rightedget}[3]{#1 \rightarrow_{#3}^{\ast} #2}

\title{Efficiently Finding All Minimal and Shortest Absent Subsequences in a String}
\authorrunning{Manea et al.}
\titlerunning{Efficiently Finding All Minimal and Shortest Absent Subsequences}

\author{Florin Manea}{Department of Computer Science, University of Göttingen, Germany}{florin.manea@cs.uni-goettingen.de}{https://orcid.org/0000-0001-6094-3324}{This work was supported by the German Research Foundation (DFG), project number 466789228 (Heisenberg grant).}

\author{Tina Ringleb}{Department of Computer Science, University of Göttingen, Germany}{tina.ringleb@stud.uni-goettingen.de}{}{}

\author{Stefan Siemer}{Department of Computer Science, University of Göttingen, Germany}{stefan.siemer@cs.uni-goettingen.de}{https://orcid.org/0000-0001-7509-8135}{}

\author{Maximilian Winkler}{Department of Computer Science, University of Göttingen, Germany}{maximilian.winkler01@stud.uni-goettingen.de}{}{}

\keywords{Subsequences, Absent Subsequences, Shortest Absent Subsequences, Minimal Absent Subsequences, Enumeration Algorithms}

\Copyright{CC-BY} 

\ccsdesc[500]{Theory of computation~Data structures design and analysis} 

\relatedversion{}

\EventEditors{John Q. Open and Joan R. Access}
\EventNoEds{2}
\EventLongTitle{42nd Conference on Very Important Topics (CVIT 2016)}
\EventShortTitle{CVIT 2016}
\EventAcronym{CVIT}
\EventYear{2016}
\EventDate{December 24--27, 2016}
\EventLocation{Little Whinging, United Kingdom}
\EventLogo{}
\SeriesVolume{42}
\ArticleNo{23}

\begin{document}

\nolinenumbers
\maketitle  

\hideLIPIcs
\begin{abstract}
Given a string $w$, another string $v$ is said to be a subsequence of $w$ if $v$ can be obtained from $w$ by removing some of its letters; on the other hand, $v$ is called an absent subsequence of $w$ if $v$ is not a subsequence of $w$. The existing literature on absent subsequences focused on understanding, for a string $w$, the set of its shortest absent subsequences (i.e., the shortest strings which are absent subsequences of $w$) and that of its minimal absent subsequences (i.e., those strings which are absent subsequences of $w$ but whose every proper subsequence occurs in $w$). Our contributions to this area of research are the following. Firstly, we present optimal algorithms (with linear time preprocessing and output-linear delay) for the enumeration of the shortest and, respectively, minimal absent subsequences. Secondly, we present optimal algorithms for the incremental enumeration of these strings with linear time preprocessing and constant delay; in this setting, we only output short edit-scripts showing how the currently enumerated string differs from the previous one. Finally, we provide an efficient algorithm for identifying a longest minimal absent subsequence of a string. All our algorithms improve the state-of-the-art results for the aforementioned problems. 
\end{abstract}


\newpage

\section{Introduction}

A string $v$ of length $\ell$ is a {\em subsequence} of a string $w$ if there exist (possibly empty) strings  $x_1, \ldots, x_{\ell+1}$ and letters $v_1, \ldots,  v_\ell$
such that $v = v_1 \cdots v_\ell$ and $w = x_1 v_1 \cdots x_\ell v_\ell x_{\ell+1}$. In other words, $v$ can be obtained from $w$ by removing some of its letters. 

The concept of subsequence (also called scattered factor, or subword) appears and plays important roles in many different areas of computer science: in combinatorics on strings~\cite{RigoS15,FreydenbergerGK15,Mat04,Salomaa05}, formal languages~\cite{simonPhD,Simon72,journals/lmcs/KarandikarS19,bib:UniReg,bib:uniCF} and automata theory~\cite{HalfonSZ17,Zetzsche16,Zetzsche18}. Aside from that, subsequences are important objects of study also in the area of algorithm-design. On the one hand, they appear in string-algorithmic problems, such as computing the {longest common subsequences} or the {shortest common supersequences} (see the survey \cite{BergrothHR00} and the references therein) or the testing of the Simon's congruence of strings (see~\cite{SimonWords,GawrychowskiEtAl2021} and the references therein). On the other hand, the combinatorial and algorithmic study of subsequences (or patterns) occurring in permutations (such as, e.g., longest increasing subsequences \cite{Fredman75}) is an important and classical topic of research (see, e.g., \cite{Kitaev:2011,Linton:2010,Romik:2015}). For a more detailed discussion regarding combinatorial pattern matching problems related to subsequences see the survey~\cite{SurveyNCMA}. Moreover, algorithmic problems on subsequences were revisited in the context of fine-grained complexity~\cite{BringmannK18,AbboudEtAl2015}, and lower bounds for their complexity were shown. Nevertheless, subsequences also appear in more applied settings: for modeling concurrency~\cite{Riddle1979a,Shaw1978}, in database theory, in settings  related to {event stream processing}~\cite{ArtikisEtAl2017,GiatrakosEtAl2020,ZhangEtAl2014}, in data mining~\cite{LiW08,LiYWL12}, or in problems related to bioinformatics~\cite{BilleEtAl2012}. These applications led to new theoretical results (algorithms and fine-grained complexity lower bounds) as well. For instance, in connection to the database connection, mentioned above, we recall the study of constrained subsequences of strings, where the substrings occurring between the positions of the subsequence are constrained by regular or length constraints~\cite{DayKMS22,AdamsonKKMS23,ManeaRS24}. \looseness=-1

This paper is centered around the notion of \emph{absent subsequence}, introduced in~\cite{KoscheKMS21}. A string $v$ is an absent subsequence of a string $w$ if $v$ is not a subsequence of $w$. The main motivation of~\cite{KoscheKMS21} for the study of absent subsequences was related to the congruence $\sim_k$ (called Simon's Congruence, see, e.g., \cite{simonPhD,Pin2019}) between strings having exactly the same set of subsequences of length at most $k$. The problems of deciding whether two given strings are $\sim_k$-equivalent, for a given $k$, and to find the largest $k$ such that two given strings are $\sim_k$-equivalent (and their applications) were heavily investigated in the literature, with optimal solutions for both these problems being obtained recently (see~\cite{GawrychowskiEtAl2021} and the references therein). One of the central ingredients in many of these works was the notion of a distinguishing string, i.e., the shortest subsequence present in one string and {\em absent} from the other. Along the same lines, and very close to the study of Simon's congruence, is the study of $m$-nearly $k$-universal strings, where one is considering strings which contain as subsequences all strings of length $k$, except exactly $m$ such strings, which are, therefore, absent \cite{FleischmannHHHMN23}. The other main motivation for investigating absent subsequences was the study of {\em missing factors} (or missing substrings) in strings, where the focus is on the set of strings which are not substrings (or factors) of $w$. The study of missing factors ranges from practical applications \cite{Chairungsee2012,Charalampopoulos2018,Crochemore2000}
to theoretical results of combinatorial \cite{SolonCPM1,Crochemore1998,Mignosi2002} or algorithmic nature \cite{SolonCPM2,Barton2014,Crochemore2020,Fujishige2016}.
Last but not least, while subsequences are used to model sequences of events occurring in a string or a stream, formalising the trace of a computation, absent subsequences correspond, similarly, to sequences of events avoided by the respective computation-trace. On a related note, absent subsequences also naturally occur in the study of patterns avoided by permutations~\cite{Kitaev:2011}, with the important difference that a permutation is essentially a string whose letters are pairwise distinct. \looseness=-1

Two classes of absent subsequences were investigated in the seminal paper~\cite{KoscheKMS21} and subsequent works \cite{KoscheKMS22,Tronicek23,abs-2407-18599}. On the one hand, we have the class of \emph{shortest absent subsequences} of a string, namely the absent subsequences of minimal length of that string. On the other hand, we have the class of \emph{minimal absent subsequences} of a string $w$: the absent subsequences whose every proper subsequence is, in fact, a subsequence of $w$. Both combinatorial and algorithmic results regarding these two notions were obtained. In \cite{KoscheKMS21,KoscheKMS22,abs-2407-18599}, combinatorial characterisations of the sets of minimal and, respectively, shortest 
absent subsequences in a string, tight bounds on the possible number of minimal or shortest absent subsequences, as well as compact representations of these sets were shown.  In \cite{KoscheKMS21,KoscheKMS22,Tronicek23}, efficient algorithms testing if a string is the shortest or minimal absent subsequence in a string were given, alongside efficient algorithms for counting and enumerating them. 

\smallskip

\noindent\textbf{Our Contribution:} We improve significantly a series of algorithmic results reported in the literature. Our main focus is on algorithms for enumerating all shortest and minimal absent subsequences of a given string of length~$n$, over an alphabet of size $\sigma$. We give optimal algorithms for these problems, having an $\O(n)$ time preprocessing and output-linear delay; the fastest known algorithms solving these problems required $\O(n\sigma)$ time preprocessing with output-linear delay (see \cite{Tronicek23}, where the results of \cite{KoscheKMS22} were improved). In a slightly different setting, we give optimal incremental enumeration algorithms (i.e., which produce, for each output-element, a short edit-script with respect to the previous output-element) for these two sets of shortest and minimal absent subsequences (with $\O(n)$ time preprocessing and $\O(1)$ delay). Finally, we give an algorithm running in $\O(n\log \sigma)$ time that computes a longest minimal absent subsequence of a string of length $n$, over an alphabet of size~$\sigma$; the fastest known algorithm solving this problem ran in $\O(n\sigma)$ time \cite{Tronicek23}. The main ingredients of our algorithms are novel and more efficient representations of minimal and shortest absent subsequences of a string, suitable for efficient enumeration, as well as a series of data-structures allowing us to efficiently search in the sets of such absent subsequences. \looseness=-1

\section{Preliminaries}\label{sec:prels}

Let $\mathbb{N}$ be the set of natural numbers. For $m, n \in \mathbb{N}$, we let $[m:n] = \{m, m+1, \ldots, n\}$ and $[n]=[1:n]$. An alphabet $\Sigma$ is a non-empty finite set of symbols called {\em letters}. A {\em string} is a finite sequence of letters from $\Sigma$; for the rest of the paper, we assume that all strings we consider are over an alphabet $\Sigma=\{1,2,\ldots,\sigma\}$. The set of all strings over $\Sigma$ is denoted by $\Sigma^*$, and let $\Sigma^+ = \Sigma^\ast \setminus \{\varepsilon\}$, where $\varepsilon$ is the empty string. The {\em length} of a string $w \in \Sigma^\ast$ is denoted by $\len w$; let $\Sigma^n = \{w\in \Sigma^\ast \mid \abs{w}=n\}$. 
The \nth{$i$} letter of $w \in \Sigma^\ast$ is denoted by  $w[i]$, for $i \in [\len w]$. For $m, n \in \mathbb{N}$, we let $w[m:n] = w[m] w[m+1] \cdots w[n]$ and $\len w_a = |\{i \in [\len w] \mid w[i] = a \}|$. A string $u=w[m:n]$ is a {\em factor} of $w$, and we have $w = xuy$ for some $x,y \in \Sigma^\ast$. If $x = \varepsilon$ (resp., $y = \varepsilon$),
$u$ is called a  {\em prefix} (resp., {\em suffix}) of $w$. Let $\al(w) = \{x \in \Sigma \mid \len w_x > 0 \}$ be the smallest subset $S \subseteq \Sigma$ such that $w \in S^\ast$. For $i=0$, we define $w^0=\varepsilon$ and, for $i>0$, we define $w^{i}=ww^{i-1}$. 
We can now recall the classical notion of subsequence.

\begin{definition}
	We call a string $v$ a \emph{subsequence} of length $|v| = k$ of string $w$,
	where $\len w = n$,
	if there exist positions $1 \leq i_1 < \ldots < i_k \leq n$,
	such that $v = w[i_1] \cdots w[i_k]$. For $j\in [k]$, we say that $v[j]$ is \emph{embedded} on position $i_j$ of $w$. 
\end{definition}

If $v$ is a subsequence of $w$, with $|v|=k$ and $|w|=n$, the \emph{canonical embedding} of $v$ in $w$ is defined as the sequence of positions $1 \leq i_1 < i_2 < \ldots < i_k \leq n$ such that $i_1$ is the leftmost occurrence of $v[1]$ in $w$, and, for $j\in [2:k]$, $i_j$ is the leftmost occurrence of $v[j]$ in $w[i_{j-1}+1:n]$. We recall the notion of $k$-universality of a string as introduced in \cite{dlt2020}: a string $w$ is \emph{$k$-universal} if every string $v$ of length $\len v \le k$  over $\al(w)$ appears as a subsequence of $w$. For a string $w $, its \emph{universality index} $\iota(w)$ is the largest integer $k$ such that $w$ is $k$-universal. Clearly, if $\iota(w) = k$, then $w$ is $\ell$-universal for all $\ell \leq k$. Note that the universality index of a string $w$ is always defined w.r.t. $\al(w)$, the alphabet of the string $w$, see \cite{dlt2020,day2021edit}.
We also recall the \emph{arch factorisation}~\cite{TCS::Hebrard1991}. For $w \in \Sigma^\ast$, with $\Sigma=\al(w)$, the {\em arch factorisation} of $w$ is defined as $w = \ar_w(1) \cdots \ar_w(k) \rest(w)$ where for all $\ell \in [k]$ the last letter of $\ar_{w}(\ell)$ occurs exactly once in $\ar_w(\ell)$, each string $\ar_w(\ell)$ is $1$-universal over $\Sigma$, and $alph(\rest(w)) \subsetneq \Sigma$. The strings $\ar_w(\ell)$ are called {\em arches} of $w$, $\rest(w)$ is called the {\em rest}. \looseness=-1

Note that every string has a unique arch factorisation and by definition each arch $\ar_w(\ell)$ from a string $w$ is $1$-universal; clearly, the number of arches in the arch factorisation of a string $w$ is exactly $\iota(w)$. By an abuse of notation, we can write $i \in \ar_w(\ell)$ if $i$ is a natural number such that $|\ar_w(1) \cdots \ar_w{(\ell-1)}| < i \leq |\ar_w(1) \cdots \ar_w{(\ell)}|$, i.e., $i$ is a position of $w$ contained in the \nth{$\ell$} arch of $w$.
The main concepts discussed in this paper were introduced in \cite{KoscheKMS22}: \looseness=-1
\begin{definition}[\cite{KoscheKMS22}] A string $v$ is an \emph{absent subsequence} of $w$ if $v$ is not a subsequence of $w$.
An {absent subsequence} $v$ of $w$ is a \emph{minimal absent subsequence} (for short, $\mas$) of $w$ if every proper subsequence of $v$ is a subsequence of $w$.
We denote the set of all $\mas$s of $w$ by $\mas(w)$.
An absent subsequence $v$ of $w$ is a \emph{shortest absent subsequence} (for short, $\sas$) of $w$ if $|v|\le |v'|$ for any other absent subsequence $v'$ of $w$.
We denote the set of all $\sas$s of $w$ by $\sas(w)$.\looseness=-1
\end{definition}

Note that any $\sas$ of $w$ has length $\iota(w)+1$ (as $w$ contains all strings of length $\iota(w)$ as subsequences, and at least one string of length $\iota(w)+1$ is not a subsequence of $w$) and $v$ is an $\mas$ of $w$ if and only if $v$ is absent and every subsequence of $v$ of length $|v|-1$ is a subsequence of $w$. For a string $w$, with $|w|=n$ and $\iota(w)=k$, and $v$ an $\sas$ of $w$, let $i_1<i_2<\ldots <i_k$ be the canonical embedding of $v[1:k]$ in $w$. Then $i_\ell\in \ar_w(\ell)$, for all $\ell\in [k]$. Moreover, $v[k+1]\notin \al(w[i_k+1:n])$, see \cite{KoscheKMS22}. 

Interestingly, every $\sas$ is an $\mas$, but not vice versa. Moreover, the length of the shortest $\mas$ of a string $w$ is $\iota(w)+1$ (as the shortest $\mas$s are the $\sas$s), but it is not obvious which is the length of a longest $\mas$ of $w$ (except that it is at most $|w|+1$). Moreover, even if the length of a longest $\mas$ is known, this still does not shed a lot of light on the set of lengths of $\mas$s. For instance, the string $1^2 2 1^5$ has $\mas$s of lengths $2,4,7,$ and $8$; these are $2^2$, $1^32$, $21^6$, and $1^8$. \looseness=-1

A series of data structures were defined in \cite{KoscheKMS22} for the efficient processing of the absent subsequences of a string; we recall some of them, which are also useful here. In the following, $w$ is a string over an alphabet $\Sigma$, with $|w|=n$ and $|\Sigma|=\sigma$. 
We firstly define the $\sigma \times n$ matrix $\nextpos[\cdot,\cdot]$, where $\nextpos[a,i]=\min (\{n+1\} \cup \{j\mid j\geq i, w[j]=a\} )$; the elements of $\nextpos[\cdot,\cdot]$ are computed by a simple dynamic programming in $\O(n\sigma)$ time, see \cite[Remark 4.2]{KoscheKMS22} for the basic idea. 
We also define the arrays $\nextArray [\cdot]$ and $\prevArray [\cdot]$, with $\nextArray [i]=\min(\{j>i\mid w[j]=w[i] \}\cup\{n+1\})$ and $\prevArray [i]=\max(\{j<i\mid w[j]=w[i] \}\cup\{0\})$; note that  $\nextArray [i]=\nextpos[w[i],i+1]$. Clearly, if $i\in [n]$ and $\nextArray[i]\neq n+1$, then $\prevArray[\nextArray[i]]=i$; if $i\in [n]$ and $\prevArray[i]\neq 0$, then $\nextArray[\prevArray[i]]=i$. The arrays $\nextArray [\cdot]$ and $\prevArray [\cdot]$ are computed in time $\O(n)$, see \cite[Lemma 4.20]{KoscheKMS22}.

We recall, as well, the array $\dist[\cdot]$, with $n$ elements, where, for $i\in [n]$, $\dist[i]=\min\{|u|\mid u$ is an absent subsequence of $w[i:n],$ starting with $w[i]\}$. This array is computed in $\O(n)$ time, see \cite[Lemma 4.16]{KoscheKMS22}. The main property of this array is the following. If $\iota(w)=k$, then, for $i\in \ar_w(\ell)$, $\ell\in[k]$, there exists an $\sas$ $v$ of $w$ such that $i$ is part of the canonical embedding of $v$ in $w$ if and only if $\dist[i]=k - (\ell-1) + 1 = k-\ell+2$ and $i$ is the first occurrence of $w[i]$ in $\ar_w(\ell)$. That is, the first $(\ell-1)$ positions, corresponding to the first $(\ell-1)$ arches of $w$, of the embedding have been chosen and we can find suitable positions for the remaining $k-(\ell-1)$ arches, as well as one position that can not be embedded in $w$.

Finally, we recall the definitions of two $k \times \sigma$ matrices $\firstPosArch[\cdot,\cdot]$ and $\lastPosArch[\cdot,\cdot]$. For $\ell \in [k]$ and $a \in \Sigma$,  $\firstPosArch[\ell,a]$ is the leftmost position of $\ar_w(\ell)$ where $a$ occurs and  $\lastPosArch[\ell,a]$ is the rightmost position of $\ar_w(\ell)$ where $a$ occurs. To avoid confusions, for all $\ell$ and $a$, the values stored in $\firstPosArch[\ell,a]$ and $\lastPosArch[\ell,a]$  are positions of $w$ (so, between $1$ and $|w|$). These two arrays are computed in $\O(n)$ time, see \cite[Lemma 4.7]{KoscheKMS22}.

We recall the following theorem from \cite{KoscheKMS21}. See \cite{KoscheKMS22} for a detailed proof.
\begin{restatable}{theorem}{mastheorem}\label{thm:mas}
Let $v,w\in \Sigma^\ast,~|v|=m+1$ and $|w|=n$, then $v$ is an $\mas$ of $w$ if and only if there exist positions $0=i_0<i_1<\ldots <i_m <i_{m+1}= n+1$ such that the following conditions are satisfied:
\begin{enumerate}
	\item 
$v=w[i_1]\cdots w[i_m]v[m+1]$; 
	\item 
$v[1]\notin \al(w[1:i_1-1])$;
	\item 
$v[r]\notin\al(w[i_{r-1}+1:i_r-1])$ for all $r\in[2: m+1]$;
	\item 
$v[r]\in \al(w[i_{r-2}+1:i_{r-1}])$ for all $r\in[2: m+1]$.\looseness=-1
\end{enumerate} 
\end{restatable}

For a string $w$, a subsequence $v$ of $w$ is called an \emph{$\mas$-prefix} of length $m$ of $w$, if $|v|=m$ and there exist positions $0=i_0<i_1<\ldots <i_m \leq n $ such that $v=w[i_1]\cdots w[i_m]$ and the following conditions are satisfied: $v[1]\notin \al(w[1:i_1-1])$ and $v[r]\notin\al(w[i_{r-1}+1:i_r-1])$ for all $r\in[2: m]$
(i.e., the sequence $i_1<\ldots <i_m$ is the canonical embedding of $v$ in $w$), and $v[r]\in \al(w[i_{r-2}+1:i_{r-1}])$ for all $r\in[2: m]$.
The following two lemmas are immediate.
\begin{restatable}{lemma}{masextensionlemma}\label{lem:MASextension}
Let $v,w\in \Sigma^\ast,~|v|=m$ and $|w|=n$, such that $v$ is an $\mas$-prefix of $w$. Then there exists an $\mas$ $u$ which has $v$ as prefix.
\end{restatable}
\begin{proof}
Let $0=i_0<i_1<\ldots <i_m \leq n $ be positions of $w$ such that $i_1<\ldots < i_m$ is the canonical embedding of $v$ in $w$. As $v$ is an $\mas$-prefix of $w$, by definition we have that $v[r]\in \al(w[i_{r-2}+1:i_{r-1}])$ for all $r\in[2: m]$. Now, let $\ell$ be the number of occurrences of $v[m]$ in $w[i_m+1:n]$. Then, $vv[m]^{\ell+1}$ is clearly an $\mas$ of $w$. 
\end{proof}

\vspace*{-0.3cm}
\begin{restatable}{lemma}{masprefixlemma}\label{lem:MASprefix}
Let $w\in \Sigma^\ast$, with $|w|=n$, $i\in [n]$, and $w[i]=a$. Then $a^{|w[1:i]|_a}$ is an $\mas$-prefix $u$ whose canonical embedding ends on position $i$.
\end{restatable}



The computational model we use to state our algorithms is the standard unit-cost word RAM with logarithmic word-size $\omega$ (meaning that each memory word can hold $\omega$ bits). It is assumed that this model allows processing inputs of size $n$, where $\omega \geq \log n$; in other words, the size $n$ of the data never exceeds (but, in the worst case, is equal to) $2^\omega$. Intuitively, the size of the memory word is determined by the processor, and larger inputs require a stronger processor (which can, of course, deal with much smaller inputs as well). Indirect addressing and basic arithmetical operations on such memory words are assumed to work in constant time. Note that numbers with $\ell$ bits are represented in $\O(\ell/\omega )$ memory words, and working with them takes time proportional to the number of memory words on which they are represented. This is a standard computational model for the analysis of algorithms, defined in \cite{FredmanW90}. Our algorithms have strings as input, so we follow a standard stringology-assumption, namely that we work with an {\em integer alphabet}: whenever we are given a string of length $n$, we assume that it is over an alphabet $\Sigma=\{1,2,\ldots,\sigma\}$, with $|\Sigma|=\sigma\leq n$. For a more detailed general discussion on the integer alphabet model see, e.\,g.,~\cite{crochemore}.

\section{Enumeration Algorithms for Paths in DAGs} \label{sec:enumSkelDAG}

We begin by giving a series of definitions related to directed acyclic graphs, which are important in our algorithms. A directed graph is a structure $G=(V, E)$, where $V$ is the set of nodes (or vertices) and $E\subseteq V\times V$ is the set of edges. The direction of the edge $(x,y)\in E$ is said to be from $x$ to $y$, $x$ is said to be the origin of $(x,y)$, and $y$ its target; we also say that $y$ is a child of $x$. A (non-trivial) walk of length $\ell$ in the directed graph $G$ is a sequence of $\ell-1$ edges $(x_1,x_2), (x_2,x_3), \ldots, (x_{\ell-1},x_\ell)$, where $\ell\geq 2$; $x_1,\ldots,x_\ell$ are the nodes on that walk, and the walk is called a path if the respective walk contains only pairwise distinct node. An $x-y$-path in $G$ is a path connecting nodes $x$ and $y$. A directed acyclic graph (DAG, for short) $G$ is a directed graph that has no cycles (i.e., if there is an $x-y$-path, then there is no $y-x$-path in $G$); note that, for DAGs, the notions of path and walk coincide and there is only a finite set of paths between two given nodes. 

Let $G=(V, E)$ be a DAG. We assume that we have a function $de: V\rightarrow E\cup\{\bot\}$ which associates to every node $x\in V$ an edge $de(x)=(x,y)$, called the {\em default edge} of $x$, or sets $de(x)=\bot$, meaning that there is no default edge leaving node $x$. Thus, we have at most $|V|$ default edges in a DAG. A path consisting of only default edges is called {\em default path}; clearly, there exists at most one default path from node $x$ to node $y$. With respect to a function $de$, defining the default edges of $G$, an arbitrary $x-y$-path $\alpha $ in $G$ can be {\em encoded}, recursively, as follows:\looseness=-1
\begin{itemize}
\item if $\alpha$ is the path of length $0$, then its encoding is $\Phi(\alpha)=\varepsilon$, or 
\item if $\alpha$ is the (unique) default $x-y$-path, then its encoding is $\Phi(\alpha) = \defPath(x,y)$, or 
\item if $\alpha$ is the sequence $ \beta_1, (z,z'), \beta_2$, where $(z,z')$ is not the default edge $de(z)$, $\beta_1$ is an $x-z$-path, and $\beta_2$ is a $z'-y$-path, then its encoding is $\Phi(\alpha)= \Phi(\beta_1)(x,y)\Phi(\beta_2)$. 
\end{itemize}
We are interested in enumerating paths between two given nodes in a given DAG. While enumeration problems and enumeration algorithms can be defined in a more general setting (see, e.g., \cite{Wasa16}), here we focus on this particular setting, to keep the presentation simple. 

Let $G=(V, E)$ be a DAG, let $de$ be a function defining the default edges of $G$, let $x,y\in V$ be two nodes of $G$, and let $A$ be the set of $x-y$-paths. Assume $|A|=r$. An enumeration of $A$ is a sequence $(s_1, s_2, \ldots, s_r)$, where $s_i\in A$, for $i\in [r]$, and $s_i\neq s_j$, for all $i,j\in [r]$, $i\neq j$. 
In this paper, we are interested in the \emph{incremental enumeration} of the set of $x-y$-paths, with respect to $de$. In that setting, instead of having explicitly the sequence $(s_1, s_2, \ldots, s_r)$ of all the paths from $A$, we have a sequence consisting of the encoding (w.r.t. $de$) of the first enumerated object $s_1$ (namely, $\Phi(s_1)$), followed by edit scripts which tell us, for $i\in [2:r]$, how to obtain the path $s_i$, the $i^{th}$ path in the enumeration, from the encoding of $s_{i-1}$, the $(i-1)^{th}$ path in the enumeration. In this paper, these edit scripts are of the form: ``return on the previous path (namely, $s_{i-1}$) to node $z$, follow some $z-y$-path (say $\gamma$, written as its encoding $\Phi(\gamma)$)''; in particular, we need to have that $s_i$ is actually the path consisting of the prefix of $s_{i-1}$ leading from $x$ to $z$ followed by the path $\gamma$.  As such, we do not have the explicit sequence of enumerated paths, but we can keep track of how the enumerated paths change from one to the next, as we advance in the sequence $(s_1, s_2, \ldots, s_r)$ from left to right. 

The \emph{enumeration problems} considered in our setting are the following: given a DAG $G=(V, E)$, with $de$ a function defining the default edges of $G$, and two nodes $x,y\in V$ of $G$, (incrementally) enumerate all $x-y$-paths of $G$. An \emph{enumeration algorithm} solving such an enumeration problem is an algorithm ${\mathcal A}$ that produces an enumeration $(s_1, s_2, \ldots, s_r)$ of the $x-y$-paths in $G$. An \emph{incremental enumeration algorithm} for the same problem, is an algorithm ${\mathcal A}$ that produces an incremental enumeration of the $x-y$-paths. Additionally, we require that our incremental enumeration algorithms can, on demand, after outputting the edit scripts that transform $s_{i-1}$ into $s_i$, also output the element $s_i$ in linear time $\O(|s_i|)$. 

Let ${\mathcal A}$ be an (incremental) enumeration algorithm for some enumeration problem, as defined above. The \emph{preprocessing time} of ${\mathcal A}$ on input $G$ is the time used by the algorithm ${\mathcal A}$ before the first output. If $(e_1, e_2, \ldots, e_r)$ is the output of ${\mathcal A}$ on input $I$ (where $e_i$ is either $s_i$, the actual $i^{th}$ path of our enumeration written explicitly, or the $i^{th}$ output in the incremental enumeration), then, for every $i \in [r]$, the \emph{delay of $e_i$} is the time between the end of completely outputting the element $e_{i-1}$ and the end of completely outputting the next output element $e_i$. The \emph{delay} of ${\mathcal A}$ on some input $G$ is the maximum over all delays of any output element. For the purpose of this paper, the preprocessing time and the delay of the algorithm ${\mathcal A}$ are the maximum preprocessing time and maximum delay, respectively, over all possible inputs $G$ with at most $n$ nodes (viewed as a function of $n$).

We note that, in the case of standard enumeration algorithms, the best we can hope for is linear preprocessing and output-linear delay (i.e., the delay between outputting $e_i$ and $e_{i+1}$ is $\O(|e_{i+1}|)$). In the case of incremental enumeration algorithms, the best we can hope for is linear preprocessing and $\O(1)$ delay. Examples of both standard and incremental algorithms can be found in \cite{Wasa16}, with some recent examples of incremental enumeration algorithms, developed in frameworks similar to ours, being \cite{adamson2024enumerating,AmarilliM23}.

In our enumeration algorithms we also assume that the inputs are not given explicitly, but in a more succinct, implicit way. We thus introduce one of the main concepts for our approach.\looseness=-1
\begin{definition}
A directed acyclic graph (for short, DAG) $G=(V,E)$ is an $m$-\emph{skeleton DAG} if the following conditions hold:
\begin{itemize}
\item There is a surjective function $\text{level} : V \rightarrow [0:m]$ such that $(x,y)\in E \!\Rightarrow\! \text{level}(x) \leq \text{level}(y)$.\looseness=-1
\item $G$ has two distinguished nodes, the source $s$ and sink $f$, with $\text{level}(s) = 0$ and $\text{level}(f) = m$. Node $s$ (resp., $f$) is the single node with $\text{level}(s) = 0$ (resp., with $\text{level}(f) = m$). \looseness=-1
\item For $\ell\in [m-1]$, the set $V_\ell = \{v\in V\mid \text{level}(v) = \ell\}$ is totally ordered by an order $\rightedge{}{}{\ell}$, henceforth called the \emph{sibling relation}. Moreover, if $V_\ell = \{v_1,\dots,v_r\}$ such that $\rightedge{v_1}{v_2}{\ell}\rightedge{}{\dots}{\ell}\rightedge{}{v_r}{\ell}$ then there is an edge $(v_i,v_{i+1})\in E$ for all $i\in[r-1]$. There are no other edges $(v,v^\prime)\in E$ such that $\text{level}(v) = \text{level}(v^\prime)$. 
We define $\rightedget{}{}{\ell}$ as the reflexive-transitive closure of $\rightedge{}{}{\ell}$. For simplicity, we denote by $\rightedge{}{}{}$ (resp., $\rightedget{}{}{}$) the union of the relations $\rightedge{}{}{\ell}$ (resp., $\rightedget{}{}{\ell}$), for all levels $\ell$.
\item For all $v\in V\setminus\{s,f\}$ there is exactly one edge $(v,v^\prime)\in E$, $v^\prime\in V$, such that $\text{level}(v) < \text{level}(v^\prime)$. The source $s$ can have multiple outgoing edges, but for any two edges $(s,v),(s,v^\prime)\in E$ we have $\text{level}(v) \neq \text{level}(v^\prime)$.
\end{itemize}
\end{definition}
Clearly, for an $m$-{skeleton DAG} $G=(V,E)$ we have that $|E|\in \O(|V|)$. We can now define the directed acyclic graph encoded by an $m$-{skeleton DAG}. See Figure \ref{Fig:DAGs}.
\begin{definition}
An $m$-skeleton DAG $G=(V,E)$ encodes the DAG $D(G) = (V, E_1 \cup E_2)$, where $E_1 = \{(v,v')\in E \mid \text{level}(v)<\text{level}(v') \}$ and $E_2 = \{(v, v'') \mid (v,v') \in E_1, \rightedget{v'}{v''}{}\}$. 
\end{definition}
Note that an $m$-skeleton DAG $G=(V,E)$ is a succinct representation of $D(G)$, as $D(G)$ might have significantly more edges than $G$ (e.g., when all the nodes of each level of $D(G)$ are connected to all the nodes on the next level, we only have two outgoing edges for each node, other than the source, in the skeleton $G$). Note that not all DAGs can be defined by skeletons (because, e.g., in a skeleton-defined DAG, there can be no edge between two children of the same node), but all trees can be defined by skeletons, and the representation of trees by skeletons is very similar to the child-sibling representation of trees \cite{FredmanSST86}. In the case of a DAG $D(G)$ defined by an $m$-skeleton DAG $G=(V,E)$, with source $s$ and sink $f$, we assume that the default edge of each node $x\in V\setminus \{s,f\}$ in $D(G)$ is the single edge $(x,y)$ of $G$ (which is also an edge of $D(G)$) with $\text{level}(y)>\text{level}(x)$, and $s$ and $f$ do not have default edges. The incremental enumeration of $s-f$-paths in $D(G)$ is done w.r.t. this definition of default paths.\looseness=-1

\begin{figure}[tb]
\centering
        \scalebox{.8}{
    \begin{minipage}{.4\textwidth}
    \centering
        \begin{tikzpicture}
    \node (s) {$s$};
    \node (2) [below =0.5cm of s] {$v_{2}$};
    \node (1) [left=0.5cm of 2] {$v_{1}$};
    \node (3) [right=0.5cm of 2] {$v_{3}$};
    \node (4) [below=.5cm of 1] {$v_{4}$};
    \node (5) [below=.5cm  of 3] {$v_{5}$};
    \node (6) [below=.5cm  of 4] {$v_{6}$};
    \node (7) [below =.5cm of 2, right=.5cm  of 6] {$v_{7}$};
    \node (8) [below=.5cm of 5] {$v_{8}$};
    \node (f) [below =.5cm  of 7] {$f$};
    \draw[->] (s) -- (1);
    \draw[->] (s) -- (4);
    \draw[->] (1) -- (5);
    \draw[->] (2) -- (7);
    \draw[->] (3) -- (5);
    \draw[->] (4) -- (f);
    \draw[->] (5) -- (6);
    \draw[->] (6) -- (f);
    \draw[->] (7) -- (f);
    \draw[->] (8) -- (f);
    \draw[->, dashed] (1) -- (2);
    \draw[->, dashed] (2) -- (3);
    \draw[->, dashed] (4) -- (5);
    \draw[->, dashed] (6) -- (7);
    \draw[->, dashed] (7) -- (8);
\end{tikzpicture}\\(a) $4$-Skeleton DAG $G$
    \end{minipage}
    \begin{minipage}{.4\textwidth}
    \centering
\begin{tikzpicture}
    \node (s) {$s$};
    \node (2) [below =.5cm of s] {$v_{2}$};
    \node (1) [left=.5cm of 2] {$v_{1}$};
    \node (3) [ right=.5cm  of 2] {$v_{3}$};
    \node (4) [below=.5cm  of 1] {$v_{4}$};
    \node (5) [below=.5cm  of 3] {$v_{5}$};
    \node (6) [below=.5cm  of 4] {$v_{6}$};
    \node (7) [below =.5cm of 2, right=.5cm  of 6] {$v_{7}$};
    \node (8) [below=.5cm of 5] {$v_{8}$};
    \node (f) [below =.5cm of 7] {$f$};
    \draw[->] (s) -- (1);
    \draw[->] (s) -- (4);
    \draw[->] (1) -- (5);
    \draw[->] (2) -- (7);
    \draw[->] (3) -- (5);
    \draw[->] (4) -- (f);
    \draw[->] (5) -- (6);
    \draw[->] (6) -- (f);
    \draw[->] (7) -- (f);
    \draw[->] (8) -- (f);
    \draw[->] (s) -- (2);
    \draw[->] (s) -- (3);
    \draw[->] (s) -- (5);
    \draw[->] (2) -- (8);
    \draw[->] (5) -- (7);
    \draw[->] (5) -- (8);
\end{tikzpicture}\\(b) Encoded DAG $D(G)$
    \end{minipage} 
    }
    \caption{$4$-Skeleton DAG $G$ (a) and its encoded DAG $D(G)$ (b), where the dashed arrows in (a) represent, respectively, the orders $\rightedge{}{}{\ell}$, for each level $\ell\in [3]$.}
    \label{Fig:DAGs}
\end{figure}
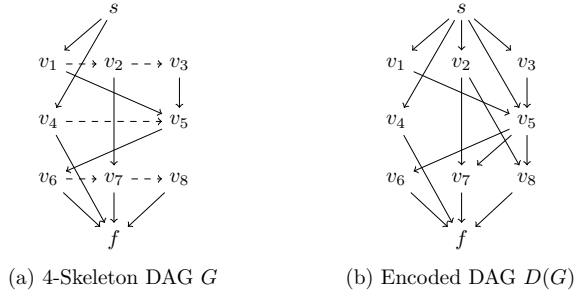

\begin{restatable}{theorem}{dagenumerationtheorem}\label{thm:DAG_enumeration}
Given an $m$-skeleton DAG $G=(V,E)$, with $|V|=n$, after an $\O(n)$ time preprocessing, we can incrementally enumerate the $s-f$-paths in $D(G)$ with $\O(1)$ delay. 
\end{restatable}
\begin{proof}
\noindent\emph{$\S $ Preliminaries.}
As a first step, we note that the case when $m=1$ is trivial (we simply have at most a single $s-f$-path, the one corresponding to the edge connecting $s$ and $f$). So we will assume in the following that $m>1$. 

Our enumeration algorithm is reminiscent of the approach presented in \cite{adamson2024enumerating}, with a main difference being that in this work we enumerate all paths in a DAG, which connect two given nodes, and might have different lengths, instead of all walks of fixed length in an arbitrary directed graph; we use the same main idea (of using default paths to succinctly represent paths in graphs), but define our algorithm iteratively (to avoid the crucial, but somehow tedious, usage of tail recursion from ~\cite{adamson2024enumerating}) and use a different toolbox of data structures. 

Before presenting our approach, we need to define the way $G$ is stored. For each node $v\in V\setminus\{s\}$, if $(v,u)\in E$ is the single edge where $\text{level}(u)>\text{level}(v)$, then we store the value $down(v)=u$. For $s$, we denote by $down(s)$ the node $u$, where $(s,u)\in E$ and $\text{level}(u)<\text{level}(x)$, for all nodes $x\in V$ such that $(s,x)\in E$. For each level $\ell$, if $V_\ell = \{v_1,\dots,v_r\}$ with $\rightedge{v_1}{v_2}{\ell}\rightedge{}{\dots}{\ell}\rightedge{}{v_r}{\ell}$, then $V_\ell$ is implemented as an array $A_{\ell}$ with $r$ elements, where $A_\ell[i]=v_i$, and denote $link(v_i)=v_{i+1}$ for $i\in [r-1]$ and $link(v_r)=\uparrow$. For each node $v=A_\ell[i]$, on level $\ell$, we store the value $pos(v)=i$. Note that we can immediately check whether there is an edge $(v,u)$ in $D(G)$: either $u=down(v)$ or $u'=down(v),$ $level(u')=level(u),$ and $pos(u')<pos(u)$. 

Recall that, for each node $v$ in $V\setminus \{s,f\}$ the {\em default edge of $v$ in $D(G)$} was defined as $de(v)=(v,down(v))$. Note that, for technical reasons, $s$ has no default edge. 
The other edges of $D(G)$ are non-default edges. As introduced before, an $x-y$-default path is an $x-y$-path consisting of default edges only and there is at most one $x-y$-default path for any $x,y\in V\setminus\{s\}$; this path is encoded as $\defPath(x,y)$. However, since every node $x\in V\setminus\{s,f\}$ has exactly one default edge, then there is exactly one $x-f$-default path for any $x\in V\setminus \{s,f\}$, namely the path starting in $x$ and obtained by following the default edges (encoded $\defPath(x,f)$). 

A node $v$ is called {\em branching node} if there is at least one non-default edge leaving $v$. Clearly, a node $v$, with the default edge $(v,v')$, is a branching node if and only if $link(v')\neq \uparrow$.

Finally, we recall that any $x-y$-path in $D(G)$, for $x,y\in V$, will be encoded as an alternating concatenation of default paths (written as their encoding) and non-default edges (originating in branching nodes, written explicitly).

Before describing the enumeration algorithm, we present a series of preprocessing steps.

\medskip

\noindent\emph{$\S$ Preprocessing.} 
Firstly, for each node $v\in V\setminus \{s,f\}$, we define the value $d(v)$ as the length of the single $v-f$-default path. The values $d(v)$, for all $v\in V$, can be computed in $\O(|V|)$ time. We simply set $d(f)=0$, and then, for $\ell=m-1$ to $1$, we compute $d(v)$ for $v\in V$ such that $\text{level}(v)=\ell$ as $d(v')+1$, where $(v,v')\in E$ is the default edge of $v$. 

Secondly, for each node $v\in V\setminus \{s\}$, we can compute the branching node $nb(v)$ which is reachable from $v$ on the single $v-f$-default path such that there is no branching node between $v$ and $nb(v)$ on the $v-f$-default path (to support the intuition, note that $nb$ stands for next branching node). We set $nb(v)=\uparrow$ if there is no branching node on the $v-f$-default path, and $nb(v)=v$ if $v$ is a branching node. The values $nb(v)$, for all $v\in V$, can be computed in $\O(|V|)$ time, similarly to the values $d(\cdot)$. We set $nb(f)=\uparrow$, and then, for $\ell=m-1$ to $1$, for $v\in V$ such that $\text{level}(v)=\ell$, we set $nb(v)=v$ if $v$ is a branching node, or, otherwise, we set $nb(v)=nb(v')$, where $(v,v')\in E$ is the default edge of $v$ (i.e., $v'=down(v)$). 

Clearly, these preprocessing steps can be implemented in linear time $\O(n)$.

\medskip

\noindent\emph{$\S$ Data Structures.} In our enumeration algorithm, we will use the notion of path-object. This is a tuple $(v,v',v'',\ell,u,\alpha )$, with $v,v',v'',u\in V\cup\{\uparrow\}$, $\ell \in [m+1]$, and $\alpha \in \{0,1\}$. Such a path-object $(v,v',v'',\ell,u,\alpha)$ encodes either the (default or non-default) edge from $v$ to $u$ (if $\ell=1$) or the default path from $v$ to $u$ (if $\ell >1$). We have $v'\neq \uparrow$ if and only if the non-default edge $(v,v')$ is in $D(G)$ and this edge was not yet explored in the enumeration. We have $v''\neq \uparrow$ if and only if $v''$ is a branching node on the default path between $v$ and $u$, and the non-default edges originating in $v''$ were not yet explored in the enumeration. Finally, the bit $\alpha$ is set to $0$ if and only if $v'=v''=\uparrow$ (and it is set to $1$, otherwise); for simplicity, the bit $\alpha$ is called the marking of the respective path-object. 

In our enumeration algorithm, we maintain two stacks: one stack ${\mathcal S}$ of path-objects, and one stack ${\mathcal C}$ which simply stores pointers to the path-objects stored in ${\mathcal S}$ whose marking is set to $1$. Very importantly, for two path-objects $x_1$ and $x_2$ whose marking equals $1$, $x_1$ is above $x_2$ in ${\mathcal S}$ if and only if the pointer to $x_1$ is above the pointer to $x_2$ in ${\mathcal C}$. In other words, the order of the objects in ${\mathcal S}$, with marking equal to $1$, is identical to the order of the pointers to these objects in ${\mathcal C}$. Intuitively, ${\mathcal C}$ keeps track of the path-objects of ${\mathcal S}$ where there is still the possibility of branching and finding paths which we did not enumerate yet. 

Our stacks are implemented statically: the stack ${\mathcal X}$, with ${\mathcal X}\in \{{\mathcal {S,C}}\}$, is an array of length $m+1$, and with a special position $top({\mathcal X})$, indicating the top of the respective stack (with $top({\mathcal X})=0$ if and only if the respective stack is empty). The typical operations are pop (which means return ${\mathcal X}[top({\mathcal X})]$ and set $top({\mathcal X})\gets top({\mathcal X})-1$) and push element $x$ (which means $top({\mathcal X})\gets top({\mathcal X})+1, {\mathcal X}[top({\mathcal X})]=x$). We will assume we can also decrease $top({\mathcal X})$ to some value $i$ (which would be equivalent to popping more than one element in $\O(1)$). Note that in all the cases when we pop elements, we do not explicitly remove them from the stack; the memory where they were stored will simply be overwritten. 

\medskip

\noindent\emph{$\S$ The Enumeration Algorithm.}
Let $(s,v)$ be one of the edges leaving $s$. We will explain how to incrementally enumerate all paths starting with $(s,v)$, with constant delay. The process is then repeated for every other edge of $G$ leaving $s$. 

Initially, the stacks ${\mathcal S}$ and ${\mathcal C}$ are empty. Firstly, we push in the stack ${\mathcal S}$ the path-objects $x_1=(s,link(v),\uparrow,1,v,\alpha)$, which encodes the edge $(s,v)$ (with $\alpha=0$ if and only if $link(v)= \uparrow$), and $x_2=(v,link(down(v)),nb(down(v)),d(v) ,f,\alpha')$ (where $\alpha'=0$ if and only if $link(down(v))=nb(down(v))=\uparrow$), which encodes the default path from $v$ to $f$. If $\alpha=1$, we insert in ${\mathcal C}$ a pointer to $x_2$. The encoding of the first path in our enumeration is now output: edge $(s,v)$, $\defPath(v,f)$ (i.e., the default path from $v$ to $f$). Then, we move on to the enumeration of the second path.

Assume now that we have enumerated $t\geq 1$ paths, and we want to enumerate path $t+1$. The procedure is the following: set the top of ${\mathcal S}$ to be the object pointed by the top of ${\mathcal C}$. Let $x=(v,v',v'',\ell,u,1)$ be the respective path-object (the current top of ${\mathcal S}$, as pointed by the top of ${\mathcal C}$). 

case (A): If $v''\neq \uparrow$, we check if $v'''=nb(down(v''))$ is still on the default path from $v$ to $u$ (i.e., we reach $u$ from $v'''$ by following the default path originating in the latter node, which is equivalent to $d(v''')\geq d(u)$). If yes, we update $x=(v,v',v''',\ell,u,1)$. If not, we update $x=(v,v',\uparrow,\ell,u,\alpha)$, where $\alpha=1$ if and only if $v'\neq \uparrow$; if $\alpha=0$, we remove the pointer to $x$ from ${\mathcal C}$. Then, with $u'=link(down(v''))$, we push in ${\mathcal S}$, in this order, the path-object $y=(v'',link(u'),\uparrow,1,u',\alpha')$, where $\alpha'=0$ if and only if $link(u')= \uparrow$, and the path-object $z=(u',link(down(u')),nb(down(u')),d(u'),f, \alpha'')$, where $\alpha''=0$ if and only if $link(down(u'))=nb(down(u'))=\uparrow$. We also push in ${\mathcal C}$, in this order, a pointer to $y$, if $\alpha'=1$, and a pointer to $z$, if $\alpha''=1$. Next, we output the edit-script showing how path $t$ is transformed in path $t+1$: return on previous path to $v$, $\defPath(v,v'')$ (i.e., default path to $v''$), non-default edge $(v'',u')$, $\defPath(u',f)$ (i.e., default path from $u'$ to $f$). 

case (B): If $v''= \uparrow$, we update $x$ to be equal to $x=(v,link(v'),\uparrow,\ell,u,\alpha)$, where $\alpha=0$ if and only if $link(v')=\uparrow$. If $\alpha=0$, we remove the pointer to $x$ from ${\mathcal C}$. Then, we push in ${\mathcal S}$ the path-object $y=(v',link(down(v')),nb(down(v')),d(v'),f, \alpha')$, where $\alpha'=0$ if and only if $link(down(v'))=nb(down(v'))=\uparrow$. We also push in ${\mathcal C}$ a pointer to $y$, if $\alpha'=1$. In this case, we output the edit-script showing how path $t$ is transformed in path $t+1$: return on previous path to $v$, non-default edge $(v,v')$, $\defPath(v',f)$ (i.e., default path from $v'$ to $f$). 

Note that the stack ${\mathcal S}$ also allows us to explicitly output the currently enumerated path in output-linear time. Assume that the first components of the path-objects stored in ${\mathcal S}$ from position $1$ to position $top({\mathcal S})$ are, in order, $s, v_1, \ldots, v_k$; let $v_{k+1}=f$. Then, we output first the edge $(s,v_1)$. Then, for $i\in [k]$ if $(v_i,v_{i+1})$ is an edge (this can be checked in $\O(1)$ time, as explained above), we output the respective edge; otherwise, we output, in order, the edges on the default path connecting $v_i$ to $v_{i+1}$. 

When the stack ${\mathcal C}$ is empty, we have finished enumerating all $s-f$-paths, which start with $(s,v)$, in $D(G)$. We restart the process for another edge leaving $s$, which was not considered yet, if such an edge exists.

\medskip

\noindent\emph{$\S$ Correctness and Complexity.}
An edge $(v,u)$ is said to explicitly appear on the stack ${\mathcal S}$ when either (A) the path-objects $y=(v,link(u),\uparrow,1,u,\alpha)$, with $\alpha=0$ if and only if $link(u)=\uparrow$, and $z = (u, link(down(u)), nb(down(u)), d(u), f, \alpha')$, with $\alpha'=0$ if and only if $nb(down(u))=link(down(u))=\uparrow$, are pushed in ${\mathcal S}$ or when (B) a path-object $x=(v,u,\uparrow ,\ell,z,1)$, which was on top of the stack, is updated to $y=(v,link(u),\uparrow,\ell,z,\alpha)$, with $\alpha=0$ if and only if $link(u)=\uparrow$, and $z=(u,link(down(u)),nb(down(u)),d(u),f, \alpha')$, with $\alpha'=0$ if and only if $nb(down(u))=link(down(u))=\uparrow$, is pushed in the stack. The time of the explicit appearance of $(v,u)$ in the stack is when (A) or (B) happens.

The insertion-prefix $p_{v,u}$ of the edge $(v,u)$ is the path associated to the path-objects stored in ${\mathcal S}$, at the time $(v,u)$ explicitly appears, and is defined as the path corresponding to the nodes appearing below (and including) the path-object $z$, at the moment of its appearance, in order from bottom to the top of that stack. This path is defined by writing the paths connecting consecutive nodes on the stack. More precisely, if $q$ and $q'$ are nodes appearing in consecutive path-objects on the stack ($q$ occurring below $q'$), then $(q,q')$ can either be the (default or non-default) edge of $D(G)$ connecting $q$ to $q'$, or the part of the $q-f$-default path starting in $q$ and ending with $q'$. 

The edge $(v,u)$ is said to be removed from ${\mathcal S}$ when one of the following happens: another edge $(v,u')$ explicitly appears on the stack, or the top of ${\mathcal C}$ is updated and points (for the first time after $(v,u)$ explicitly appeared in ${\mathcal S}$) to an element below $y$. 

It can be shown by induction on $d(u)$ that between the step when an edge $(v,u)$ appears in the stack, and the step when it is removed, all paths starting with $p_{v,u}$ followed by $u-f$-paths are incrementally enumerated.

The base case is $d(u)=1$ (as $f$ is the only node with $d(\cdot)=0$, and no edge $(v,f)$ can appear explicitly in the stack), and it is immediate. Basically, once $(v,u)$ appears, we enumerate the single $u-f$-path, corresponding to the single edge leaving $u$ (towards $f$), preceded by the path corresponding to $p_{v,u}$.

Let us now assume that the statement holds whenever $d(u)\leq k$, and show that it also holds for $d(u)=k+1$. Assume, therefore, that an edge $(v,u)$ appears in the stack, with $d(u)=k+1$. The first path that we enumerate is that corresponding to $p_{v,u}$ followed by the default path from $u$ to $f$. Then, we find the first branching node $v'$ on this default path from $u$ to $f$, and explore the edge $(v',v'')$, where $v''=link(down(v'))$; this means that the respective edge explicitly appears in the stack. By induction, all paths corresponding to $p_{v',v''}$ (i.e., $p_{v,u}$, followed by the path from $u$ to $v'$, and the edge from $v'$ to $v''$), completed by $v''-f$-paths, are enumerated now. This process is repeated for all branching nodes on the path from $u$ to $f$. After all these are enumerated, the edge $(u,u')$ is inserted, where $u'=link(down(u))$ (if such an edge exists). By induction, all paths corresponding to $p_{v,u}$, followed by the edge $(u,u')$, and then all $u'-f$-paths are enumerated. The process is now repeated for all nodes $u''\neq u'$ such that $\rightedget{u'}{u''}{\ell}$, with $\ell=level(u')=level(u'')$. Putting all these together, it follows that the statement holds for $u$ (all paths corresponding to $p_{v,u}$, followed by $u-f$-paths are correctly enumerated). 

Based on the above, between the point when an edge $(s,v)$ appears on the stack, and the point it is removed, we incrementally enumerate all paths starting with $(s,v)$ followed by $v-f$-paths. So, our algorithm works correctly.

As for the complexity of our algorithm: the preprocessing clearly takes linear time. Then, the time between completely producing the outputs corresponding to two consecutive paths is also, clearly, constant: we use ${\mathcal C}$ to determine and move to the new top of ${\mathcal S}$, in $\O(1)$ time; then we update the top element of ${\mathcal S}$ in $\O(1)$ time, and insert at most two new elements on this stack; then we produce the output, which consists of a sequence of at most three modifications of the previous path. In conclusion, everything still takes $\O(1)$ time. This concludes the proof that our incremental enumeration algorithm works with $\O(1)$ delay, after a linear time preprocessing.
\end{proof}

Worth noting, our incremental enumeration algorithm can be easily transformed to also rank and unrank paths in graphs, just as for the algorithm of \cite{adamson2024enumerating}. Moreover, it also can be used to count the $s-f$-paths in $D(G)$; this can be done in $\O(|V|)$ steps by using a bottom-up strategy computing the number of $v-f$-paths originating in every node, and also computing, for each node $u$, the number of $x-f$-paths originating in nodes $\rightedget{u}{x}{}$ with $level(u)=level(x)$ (as soon as we computed the number of paths originating in every node of that level). 
Such applications are not in the focus of this paper; we only mention them as they support the usefulness of our approach beyond the direct applications in the enumeration of $\mas$s and $\sas$s. \looseness=-1
One obtains the following result immediately.
\begin{restatable}{corollary}{dagstandardenum}\label{thm:DAG_standard_enumeration}
Given an $m$-skeleton DAG $G=(V,E)$, we can preprocess it in $\O(|V|)$ time and enumerate the $s-f$-paths in $D(G)$ with output-linear delay. 
\end{restatable}
\begin{proof}
By the definition of incremental enumeration algorithms, we can output explicitly the currently enumerated object, in linear time w.r.t. its length, at every step of the enumeration. In the proof of Theorem \ref{thm:DAG_enumeration}, we have explained how the currently enumerated path can be obtained, in linear time w.r.t. its length, from the maintained data structures. 
\end{proof}

\section{Enumerating Shortest Absent Subsequences Optimally}\label{sec:SASMAS}


We firstly aim to apply Theorem \ref{thm:DAG_enumeration} for the enumeration of $\sas(w)$.
\begin{restatable}{theorem}{sasskeletontheorem}\label{thm:SAS_skeleton}
    For a string $w\in \Sigma^n$, with $\iota(w)=k$, we can construct in $\O(n)$ time a $(k+2)$-skeleton DAG $G$, with $\O(n)$ nodes, source $s$ and sink $f$, such that $\sas(w)$ can be bijectively mapped to the $s-f$-paths of $D(G)$.
\end{restatable}

\begin{proof}
\noindent \emph{$ \S$ Preliminaries.}
We will make use of the array $\dist[\cdot]$, with $n$ elements, as well as the arrays  $\firstPosArch[\cdot,\cdot]$ and $\lastPosArch[\cdot,\cdot]$ introduced in the preliminaries. For simplicity of the exposure, let $F_\ell=\{i\in \ar_w(\ell)\mid i=\firstPosArch[\ell,w[i]]$ and $\dist[i]=k-\ell+2\}$, and let $F=\bigcup_{\ell=1}^k F_\ell$. 
Recall that, for $\ell\in [k]$, if $i\in \ar_w(\ell)$ then there exists an $\sas$ $v$, with $|v|=k+1$, of $w$ such that $i$ is part of the canonical embedding of $v[1:k]$ in $w$ if and only if $i\in F_\ell$. 

Also, for simplicity of exposure, we use the notation $w[n+a]=a$, for all letters $a\in [\sigma]=\Sigma$.

Let $V=\{s,f\} \cup F \cup \{n+i\mid i\in [\sigma]\}$. Let $level:V\rightarrow (\{0\}\cup [k+2])$ be a function such that $level(s)=0$, $level(f)=k+2$, $level(n+i)=k+1$ for all $i\in \Sigma$, and, for $i\in F_\ell$, $level(i)=\ell$. Now, let $E = \{(s,i)\mid i\in F, level(i)=1\} \cup \{(n+i,f)\mid i\in [\sigma]\} \cup \{(i,j)\mid i,j\in F, level(i)=level(j)-1, \nextpos[w[j],i+1]=j\} \cup \{(i,n+j)\mid i\in V, level(i)=k+1, j\notin \al(w[i+1:n])\}$. 

Note that, for $\ell\in [k]$ and $i\in F_{\ell}$, $j\in F_{\ell+1}$, we have that $(i,j)\in E$ if and only if $\lastPosArch[\ell, w[j]]\leq i$. 

We define the DAG $K=(V,E)$, and claim that there is a bijection between $s-f$-paths in $K$ and $\sas(w)$. Let us first consider an $\sas$ $v$ of $w$ and let $i_1 < \ldots <i_k$ be the canonical embedding of $v[1:k]$ (recall that $|v|=k+1$, and also that $i_j\in \ar_w(j)$ for all $j\in [k]$). Then, the following hold:
\begin{itemize}
    \item $i_j\in F_j$ for all $j\in [k]$, as $i_j$ is part of a canonical embedding of an $\sas$, and $i_j\in \ar_w(j)$ (so, $level(i_j)=j$). Thus, $i_1\in F_1$ and $(s,i_1)\in E$.
    \item $i_{j}=\nextpos[w[i_j],i_{j-1}+1]$, for all $j\in [k]$, with $i_0=0$, as $i_1<\ldots<i_k$ is the canonical embedding of $v[1:k]$. Thus, $(i_j,i_{j+1})\in E$, for all $j\in [k-1]$. 
    \item $v[k+1]\notin \al(w[i_k+1:n])$. Thus, $(i_k,n+v[k+1])\in E$. 
\end{itemize}
Therefore, we can map the $\sas$ $v$ to a path in the DAG $K$, namely the one corresponding to the nodes $s,i_1,\ldots,i_k,n+v[k+1],f$. This mapping is clearly injective. 

Now consider an $s-f$-path in $K$. Assume that the nodes on this path are $s,i_1,\ldots,i_k,n+v[k+1],f$. This can be injectively mapped to the string $v=w[i_1]w[i_2]\cdots w[i_k]v[k+1]$. By the definition of $E$, we have that $i_j=\nextpos[w[i_j],i_{j-1}+1]$, so $i_1<\ldots <i_k$ is the canonical embedding of $v[1:k]$, and $v[k+1]\notin \al(w[i_k+1:n])$. Therefore, $v$ is an $\sas$ of $w$. 

So, at this point, we have a DAG $K$, with two nodes $s$ and $f$, whose $s-f$-paths correspond bijectively to $\sas(w)$. In the following, we show that we can construct in linear time $\O(n)$ a $(k+2)$-skeleton DAG $G$ such that $D(G)=K$. 

We begin by noting that there are at most $\sigma$ nodes $i\in V$ such that $level(i)=\ell$, for all $\ell \in  \{0\}\cup [k+2]$. In other words, each level of $K$ contains at most $\sigma$ nodes. As $n=|w|\geq \iota(w) \sigma = k\sigma$, we get that the number of nodes of $K$ is $\O(n)$. Let us first define, for each $i\in F$, the set $C(i)=\{j\mid (i,j)\in E\}$. We have that, if $level(i)=\ell<k$, then $C(i)=\{j\mid j\in F_{\ell+1}, \lastPosArch[\ell,w[j]]\leq i\}$, and if $level(i)=k$, then $C(i) = \{n+a\mid a\in \Sigma\setminus\al(r(w)), \lastPosArch[k,a]\leq i\}$. We note that if $level(i)=level(i')$ and $i<i'$ then $C(i)\subseteq C(i')$; indeed, if $level(i)=level(i')=\ell$, $j=\nextpos[w[j],i+1]$ for some $j$ with $level(j)=\ell+1$, then $j=\nextpos[w[j],i'+1]$.

\medskip

\noindent \emph{$ \S$ Skeleton Construction.} 
The algorithm constructing the $(k+2)$-skeleton DAG $G$ starts with a linear time preprocessing phase in which we construct the arch factorisation of $w$, as well as the arrays $\dist[\cdot]$, $\firstPosArch[\cdot,\cdot]$, and $\lastPosArch[\cdot,\cdot]$. Further, we go through the letters of each arch $\ell$, for $\ell \in [k]$, and mark the positions $i$ such that $i=\firstPosArch[\ell,w[i]]$ and $\dist[i]=k-\ell+2$. These marked positions are the elements of $F$, and we can extract the increasingly sorted list $F_\ell$ of the positions of $F$ which are part of $\ar_w(\ell)$, for all $\ell\in [k]$. Similarly, we go through the letters of each arch $\ell$, for $\ell \in [k]$, and mark (with a different marking as before) the positions $i$ such that $i=\lastPosArch[\ell,w[i]]$. We can now extract the increasingly sorted list $G_\ell$ of the positions $\lastPosArch[\ell,a]$, for $a\in \Sigma$. 

When $\ell<k$, we remove from $G_\ell$ the positions $i$ such that $\dist[\firstPosArch[\ell+1,w[i]]]\neq k-(\ell+1)+2$. When $\ell=k$, we remove from $G_\ell$ the positions $i$ such that $w[i]$ occurs in $r(w)$ (the rest of $w$). 

We can now define the $(k+2)$-skeleton DAG $G=(V,E')$. The set of nodes is $V$, and the $level$ function is the one defined above for the graph $K$. By definition, the nodes on level $\ell$ are the elements of the set $F_\ell$. 

We further need to define the order $\rightedge{}{}{\ell}$ for each level $\ell\in [k+1]$ (and, as such, the edges between the nodes on the same level). On level $1$, $\rightedge{}{}{1}$ coincides with the inverse natural order on integers (i.e., $\rightedge{i}{j}{1}$ if and only if $i > j$); note that the order of these nodes does not really play any role in our algorithms. On level $\ell$, for $\ell\in [2:k+1]$, the ordering is the one induced by the reversed order of the elements in $G_{\ell-1}$: for $r,q\in F_{\ell}$, $\rightedge{r}{q}{\ell}$ (or, in other words, there is a path from $r$ to $q$ on level $\ell$ of $G$) if and only if $\lastPosArch[\ell-1,w[r]]>\lastPosArch[\ell-1,w[q]]$. Equivalently, for $\ell\leq k$, we have $\rightedge{\firstPosArch[\ell,w[r]]}{\firstPosArch[\ell,w[q]]}{\ell}$ if and only if $\lastPosArch[\ell-1,w[r]]$ comes after $\lastPosArch[\ell-1,w[q]]$ in $G_{\ell-1}$. The ordering on each level induces the edges of $E'$ between nodes with the same level.

Before defining the edges of $E'$ connecting nodes on different levels, recall that, for some $i\in F_\ell$, $C(i)=\{j\in F_{\ell+1} \mid \lastPosArch[\ell,w[j]] \leq i\}$. Thus, for node $i$ on level $\ell\in [k]$, the single edge $(i,j)$ with $level(j)>level(i)$ is defined as follows: let $j' =\max \{d\in G_{\ell}\mid i\geq d\}$ and then set $j=\firstPosArch[\ell+1,w[j']]$ (if $\ell<k$) or $j=n+w[j']$ (if $\ell=k$). Moreover, there is an edge from $s$ to the smallest node $\min\{v\in F_1\}$ w.r.t. $\rightedge{}{}{1}$ on level $1$, and there is one edge from each node on level $k+1$ to $f$. 

At this point it is important to note that for some $(i,j)\in E'$, where $level(i)<level(j)$, we have that $r\in C(i)$ if and only if $\rightedge{j}{r}{\ell+1}$ (that is, $r$ is accessible from $j$ on level $level(j)$). It is now immediate that $G$ is a $(k+2)$-skeleton DAG, and that $D(G)=K$. 

\medskip

\noindent \emph{$ \S$ Complexity.} All components of $G$ can be computed in linear time in a straightforward manner, except the set of edges connecting nodes on different levels. 

These edges are computed by the following procedure. For each $\ell \in [k]$, let $d$ be the first element in $G_\ell$ and then, for each node of $i\in F_\ell$ (recall that $F_\ell$ is a list sorted increasingly) execute the following process. Go through $G_\ell$, starting from element $d$, until we reach an element $j''>i$; set $d=j''$, let $j'$ be the predecessor of $j''$ in $G_\ell$, and define the edge $(i,\firstPosArch[\ell+1,w[j']])$ (if $\ell+1\leq k$) or the edge $(i, n + w[j'])$ (if $\ell + 1= k+1$); then repeat this for the next $i\in F_\ell$. 

This process takes $\O(|F_\ell|+|G_\ell|)$ time to compute the edges going from level $\ell$ to level $\ell+1$. Summing up over all levels, we obtain that, overall, computing the set of edges connecting nodes on different levels of $G$ takes linear time $\O(n)$. 

So, in conclusion, $G$ can be computed in linear time from $w$, and the $s-f$-paths of $D(G)$ correspond bijectively to the $\sas(w)$. 

This concludes our proof.
\end{proof}

Two exemplifying figures (Figures~\ref{fig:massas1}, \ref{fig:massas2}) as well as pseudocode implementing its construction (Algorithm~\ref{SAS:algo:main}) are given in Appendix~\ref{sec:figures}. By Theorems \ref{thm:DAG_enumeration},  \ref{thm:DAG_standard_enumeration}, and \ref{thm:SAS_skeleton}, we obtain the following results:\looseness=-1
\begin{restatable}{theorem}{sasenumtheorem}\label{thm:SAS_enumeration}
    Given a string $w\in \Sigma^n$, after an $\O(n)$ time preprocessing, we can incrementally enumerate $\sas(w)$ with $\O(1)$ delay.
\end{restatable}

\begin{corollary}\label{cor:SAS_standard_enumeration}
    Given a string $w\in \Sigma^n$, after an $\O(n)$ time preprocessing, we can enumerate $\sas(w)$ with output-linear delay.
\end{corollary}

\section{Enumerating Minimal Absent Subsequences Optimally}\label{sec:SASMAS2}

Firstly, we show that the strategy used to enumerate $\sas(w)$ can be used to enumerate $\mas(w)$, and obtain a result similar to, but not as efficient as, Theorem \ref{thm:SAS_enumeration}.

\begin{restatable}{theorem}{masskeletontheorem}\label{thm:MAS_skeleton}
    For a string $w\in \Sigma^n$, over an alphabet of size $\sigma$, we can construct in $\O(n\sigma)$ time an $(n+1)$-skeleton DAG $G$, with $\O(n\sigma)$ nodes, source $s$, sink $f$, such that $\mas(w)$ can be bijectively mapped to the set of $s-f$-paths of $D(G)$.
\end{restatable}

\begin{proof}
\noindent\emph{$\S $ Preliminaries.}
    Before constructing the skeleton graph $G$, we introduce for each position $i\in[n]$ of $w$ an ordered set $P_i = \{j_{i,1} < j_{i,2}< \ldots < j_{i,n_i}\}$, where $n_i = \abs{\text{alph}(w[1:i])}$, such that for each $\ell\in[n_i]$ holds that $j_{i,\ell}\in P_i$ if and only if $w[j_{i,\ell}]\in\text{alph}(w[j_{i,\ell}:i])$ and $w[j_{i,\ell}]\notin\text{alph}(w[j_{i,\ell}+1:i])$, i.e., $w[j_{i,\ell}:i]$ is the smallest substring ending in position $i$ that contains letter $w[j_{i,\ell}]$. 
 
\medskip

\noindent\emph{$\S $ Construction of the skeleton graph.}
    We will now define the skeleton graph $G = (V,E)$. We create a set of nodes $V = \{s,f\}\cup\{ V_i \mid i\in[n] \}$, with $ V_i = \{ v_{i}^{\ell} \mid \ell\in[n_i]\}$ corresponding to positions from $P_i$ ($ v_{i}^{\ell}$ corresponds to the pair $(i,j_{i,\ell})$). We construct the edges $(v,v') \in E$, as follows: 
    \begin{itemize}
        \item There is an edge $(s,v_{i}^{1})$, if $i = \nextpos[a,1]$. 
        \item For $i\in[n]$ and $\ell\in[n_i]$ there is an edge $(v_{i}^{\ell},v_{k}^{\ell^\prime})$, if $k = \nextpos[w[j_{i,\ell}],i+1] \leq n$ and $\ell^\prime = \arg\min_{t}\{j_{k,t}\in P_k \mid j_{k,t} > i\}$.
        \item For $i\in[n]$ and $\ell\in[n_i]$ we have an edge $(v_{i}^{\ell},f)$, if $w[j_{i,\ell}]\notin\al(w[i+1:n])$. 
        \item For all $i\in[n]$ and $\ell\in[n_i-1]$ there is an edge $(v_{i}^{\ell}, v_{i}^{\ell+1})$.
    \end{itemize}
    
    We define $\text{level}(s) = 0, \text{level}(f) = n+1$, and for all $i\in[n],\ \ell\in[n_i],$ we set $\text{level}( v_{i}^{\ell} ) = i$.
    The set $P_i = \{j_{i,1} < \ldots < j_{i,n_i}\}$ trivially induces an ordering $\rightedge{}{}{i}$ on the nodes $v_{i}^{\ell}$ such that $\rightedge{v_{i}^{1}}{v_{i}^{2}}{i}\rightedge{}{\ldots}{i}\rightedge{}{v_{i}^{n_i}}{i}$ for all $i\in[n]$. 

    By the definition of $V_i$ the number of nodes in each set $V_i$ is upper bounded by $\abs{V_i} \leq \sigma$. This bound on each of the sets yields the overall upper bound $\abs{V} \leq n\sigma + 2\in\O(n\sigma)$ on the full set of nodes. Further, every node, except for $s$ and $f$, has at most two outgoing edges and the number of edges originating in $s$ is also bounded by $\sigma$. Therefore, $\abs{E}\in\O(n\sigma)$ as well.  

\medskip

\noindent\emph{$\S $ Bijection.}
    We claim that the strings in $\mas(w)$ correspond bijectively to the $s-f$-paths in $D(G)$. 

    Let $u$ be an MAS of $w$, with $\abs{u} = m+1$. We recall that by Theorem~\ref{thm:mas}, there exist positions $0 = i_0 < i_1 < \ldots < i_m < i_{m+1} = n+1$ such that the following four conditions are satisfied: \begin{enumerate}
        \item $u = w[i_1]\cdots w[i_m]u[m+1]$,
        \item $u[1] \notin \al(w[1:i_1-1])$,
        \item $u[k] \notin \al(w[i_{k-1}+1:i_k-1])$ for all $k\in[2:m+1]$,
        \item $u[k] \in \al(w[i_{k-2}+1:i_{k-1}])$ for all $k\in[2:m+1]$.
    \end{enumerate}
    
    We will show by induction that $u$ corresponds to a path $s, v_{i_1}^{\ell_1}, v_{i_2}^{\ell_2}, \ldots, v_{i_m}^{\ell_m}, f$. Since $i_1 = \nextpos[u[1], 1]$ and $w[i_2]\in w[1:i_1]$, there exists an $\ell_1\in [n_{i_1}]$ such that $w[j_{i_1,\ell_1}] = w[i_2]$, and thus there is an edge $(s, v_{i_1}^{\ell_1})$. 

    Assume now that the path $s, v_{i_1}^{\ell_1}, v_{i_2}^{\ell_2}, \ldots, v_{i_k}^{\ell_k}$ corresponds to the MAS-prefix $w[i_1]\cdots w[i_k]$ for $k\in [m-1]$ and $w[j_{i_k,\ell_k}] = w[i_{k+1}]=u[k+1]$, then the path $s, v_{i_1}^{\ell_1}, \ldots, v_{i_{k+1}}^{\ell_{k+1}}$ corresponds to the MAS-prefix $w[i_1]\cdots w[i_{k+1}]$ as follows. According to the induction hypothesis, $w[j_{i_k,\ell_k}] = w[i_{k+1}]$ and, according to the definition of $D(G)$, every edge originating in $v_{i_k}^{\ell_k}$ is of the form $(v_{i_k}^{\ell_k}, v_{i_{k+1}}^\ell)$, where $\ell\in\{t \in [n_{i_{k+1}}] \mid j_{i_{k+1}, t} > i_k\}$, since $i_{k+1} = \nextpos[w[i_{k+1}], i_k+1]$. By Theorem~\ref{thm:mas}, $u[k+2] \in w[i_k+1 : i_{k+1}]$, and, according to the definition of $P_{i_{k+1}}$, there exists an $\ell^\prime\in[n_{i_{k+1}}]$ such that $j_{i_{k+1},\ell^\prime} > i_{k}$ and $\al(w[i_k+1:i_{k+1}]) = \al(w[j_{i_{k+1,\ell^\prime}} : i_{k+1}])$. Therefore, we can extend the path $s, v_{i_1}^{\ell_1}, \ldots, v_{i_k}^{\ell_k}$ by node $v_{i_{k+1}}^{\ell_{k+1}}$ for $\ell_{k+1} \in [\ell^\prime : n_{i_{k+1}}]$ such that $w[j_{i_{k+1}, \ell_{k+1}}] = u[k+2]$. 

    Now, if the path $s, v_{i_1}^{\ell_1}, \ldots, v_{i_m}^{\ell_m}$ corresponds to the MAS-prefix $w[i_1]\cdots w[i_m]$ and $w[j_{i_m,\ell_m}] = u[m+1]$, then we can extend this path corresponding to the MAS $u$ by appending the sink $f$: since $u$ is an MAS, $\nextpos[w[j_{i_m,\ell_m}], i_m+1] = \nextpos[u[m+1], i_m+1] = n+1$ holds and thus, the only edge originating in $v_{i_m}^{\ell_m}$ is the edge $(v_{i_m}^{\ell_m}, f)$. We can therefore extend the path $s, v_{i_1}^{\ell_1}, \ldots, v_{i_m}^{\ell_m}$ using this edge, which yields a complete $s-f$-path in $D(G)$. Thus, every MAS has a corresponding $s-f$-path in $D(G)$. 

    For the reverse implication, consider a path from $s$ to $f$ in $D(G)$, namely $s, v_{i_1}^{\ell_1}, v_{i_2}^{\ell_2},\ldots, v_{i_m}^{\ell_m}, f$, where $w[j_{i_k,\ell_k}] = w[i_{k+1}]$ holds for all $k\in[m-1]$. We claim that this path corresponds to the MAS $u = w[i_1]w[i_2]\ldots w[i_m]b$, where $b = w[j_{i_m,\ell_m}]$. 
    According to the definition of $D(G)$, there is an edge $(s, v_{i_1}^{\ell_1})$ if and only if $i_1 = \nextpos[a, 1]$ for some $a=w[i_1]\in\Sigma$. Further, since $(v_{i_k}^{\ell_k}, v_{i_{k+1}}^{\ell_{k+1}})$ is a valid edge if and only if $i_{k+1} = \nextpos[w[j_{i_k,\ell_k}], i_k+1]$ and $j_{i_{k+1},\ell_{k+1}} > i_k$ for all $k\in[m-1]$, $w[i_{k+1}] = w[j_{i_k,\ell_k}] \in \al(w[j_{i_k,\ell_k}: i_k]) \subseteq \al(w[i_{k-1}+1 : i_k])$ and $w[i_{k+1}] \notin \al(w[i_k+1:i_{k+1}-1])$ hold and therefore, conditions 3 and 4 hold for all $k\in[m-1]$. By the same argument, conditions 3 and 4 also hold for $u[m+1] = b$, as $b\notin \al(w[i_m+1:n])$ (since there is an edge $(v_{i_m}^{\ell_m}, f)$) and $b = w[j_{i_m,\ell_m}] \in \al(w[j_{i_m,\ell_m}:i_m]) \subseteq \al(w[i_{m-1} : i_m])$ hold. Therefore, $u$ fulfills the four conditions 1 - 4 and every $s-f$-path in $D(G)$ corresponds to an MAS of $w$.

\medskip

\noindent\emph{$\S $ Complexity.}
    We will first show that the sets $\{P_1, \dots, P_n\}$ can be precomputed in $\mathcal O(n\sigma)$ time and space. 

    An important initial observation is that the sets $P_i = \{j_{i,1},\dots, j_{i,n_i}\}$ and $P_{i+1} = \{j_{i+1,1}, \dots, j_{i+1, n_{i+1}}\}$ differ in at most two elements $j_{i,\ell}$ and $j_{i+1,\ell^\prime}$. Namely, $P_{i+1}$ must contain $i+1 = j_{i+1,n_{i+1}}$ (since $w[i+1:i+1]$ is the smallest interval containing $w[i+1]$) and $P_i$ can not contain $i+1$. 
    If position $i+1$ is the first occurrence of letter $w[i+1]$, then $n_i < n_{i+1}$ and $P_i \subsetneq P_{i+1} = P_i \cup \{i+1\}$. 
    Otherwise, $P_i$ contains an element $j_{i,\ell}$ such that $w[i+1] = w[j_{i,\ell}]$. Since $\text{alph}(w[j_{i,\ell} : i+1]) \supseteq \text{alph}(w[i+1:i+1]) = w[i+1]$, $w[j_{i,\ell}:i+1]$ is not the smallest interval ending in $i+1$ and containing $w[i+1]$, thus, $j_{i,\ell}\notin P_{i+1}$. 
    All other elements $j\in P_i$ are contained in $P_{i+1}$ as well. 
        
    As illustrated in Algorithm~\ref{MAS:algo:KM}, we can maintain a doubly linked list of the latest occurrence of each letter $a\in\Sigma$ up to position $i\in[n]$. In each step, we update letter $w[i]$ and copy the contents of the doubly linked list into row $i$ of an $n\times \sigma$ matrix $\KM$. 
    Since each update uses a constant number of steps and copying takes $\mathcal O(\sigma)$ time, we achieve an overall runtime of $\mathcal O(n\sigma)$ while using data structures of size $\mathcal O(n\sigma)$.  

    Using the sets $\{P_1,\dots,P_n\}$ we can now show that the computation of the edges $(v_i^\ell, v_k^{\ell^\prime})$ and $(v_i^\ell, f)$ for all $i\in[n]$ and $\ell\in[n_i]$, where $k = \nextpos[w[j_{i,\ell}],i+1]$ and $\ell^\prime = \min\{t\in[n_k]\mid j_{k,t} > i\}$ (if $k \leq n$), can be done in $\mathcal O(n\sigma)$ time as well. 

    We observe that, given two nodes $v_{x}^{y}$ and $v_{x^\prime}^{y^\prime}$, $j_{x,y}\in P_x, j_{x^\prime,y^\prime}\in P_{x^\prime}$, such that $w[j_{x,y}] = w[j_{x^\prime,y^\prime}]$, $x^\prime < x$, and $\nextpos[w[j_{x,y}],x+1] = \nextpos[w[j_{x^\prime,y^\prime}], x^\prime+1] = i$ (i.e., $v_x^y$ and $v_{x^\prime}^{y^\prime}$ both lead to a node on level $i$), then $j = \min\{p\in P_i\mid p > x\}$ is at least as large as $j^\prime = \min\{p\in P_i\mid p > x^\prime\}$. Thus, we can restrict the search for $j$ to the set $\{p\in P_i \mid p \geq j^\prime\} \subseteq P_i$. 

    We can therefore iteratively build the graph $G$ by computing all outgoing edges of each level $i\in[n]$, from the smallest to the largest level.
    To achieve this, we maintain information about the set restrictions mentioned above using an array $S$ of length $n$. Intuitively, $S[k]$ stores, for each value $k\in[n]$, the index $\ell^\prime$, such that $(v_i^\ell, v_k^{\ell^\prime})$ is the edge to level $k$ that was last added, or $1$, if no edge to level $k$ has been added to $G$ yet; note that $S[k] = 1$ is also possible if $i < j_{k,1}$.
    We initialize the values of $S$ to be $1$. 
    
    In each step, we obtain the outgoing edge of $v_{i}^{\ell}$ as follows.
    If $w[j_{i,\ell}]\notin \text{alph}(w[i+1:n])$, we simply add $(v_i^\ell, f)$ to the edge set. Assume now that $k = \nextpos[w[j_{i,\ell}], i+1] \leq n$ and set $s = S[k]$. 
    If $j_{k,s} > i$, we already found the smallest possible $j = j_{k,s} > i$ and add the resulting edge $(v_{i}^{\ell},v_{k}^{s})$.
    Otherwise, we iteratively search for $\ell^\prime\in[s+1:n_k]$ such that $j_{k,\ell^\prime} > i$ and $\ell^\prime$ is minimal,
    and add the resulting edge $(v_{i}^{\ell},v_{k}^{\ell^\prime})$. 
    We then set $S[k] = \ell^\prime$ and continue with the next node $v_i^{\ell+1}$ (or $v_{i+1}^1$, if $\ell=n_i$). 

    The value of $S[i]$, for $i\in[n]$, is updated at most $n_i-1 < \sigma$ times, where each update from value $s$ to $s^\prime$ takes $\mathcal O(\abs{s^\prime-s})$ time. 
    Therefore, updating all values $S[i]$ during the computation of the $\mathcal O(n\sigma)$ edges takes $\mathcal O(n\sigma)$ time. 
    Additionally, in each step we query the value of $S[i]$, which takes constant time, and use $\nextpos[a,i]$, which can be done in constant time after $\mathcal O(n\sigma)$ time preprocessing. 
    Thus, the computation of $(v_{i}^{\ell},v_{k}^{\ell^\prime})$ for all $i\in[n],~ \ell\in[n_i],~ k = \nextpos[w[j_{i,\ell}], i+1],~ \ell^\prime = \arg\min_{t}\{j_{k,t}\in P_k\mid j_{k,t} > i\}$, and therefore also the complete construction of the $(n+1)$-skeleton DAG $G$, can be done in $\mathcal O(n\sigma)$ time and space, which concludes the proof of our statement. 
\end{proof}

Two exemplifying figures (Figures~\ref{fig:massas1}, \ref{fig:massas2}) as well as pseudocodes implementing the computation of the sets $P_i$ (Algorithm~\ref{MAS:algo:KM}) and the construction of the skeleton-DAG (Algorithm~\ref{MAS:algo:main}) are given in Appendix~\ref{sec:figures}.
Now, by Theorems \ref{thm:DAG_enumeration} and \ref{thm:DAG_standard_enumeration}, we obtain the following results: 

\begin{restatable}{theorem}{masenumtheorem}\label{thm:MAS_enumeration}
    Given a string $w\in \Sigma^n$, over an alphabet of size $\sigma$, after an $\O(n\sigma)$ time preprocessing, we can incrementally enumerate $\mas(w)$ with $\O(1)$ delay.\looseness=-1
\end{restatable}

\begin{corollary}\label{cor:MAS_standard_enumeration}
    Given a string $w \in \Sigma^ n$, over an alphabet of size $\sigma$, after an $\O(n\sigma)$ time preprocessing, we can enumerate $\mas(w)$ with output-linear delay.
\end{corollary}

The result of Corollary \ref{cor:MAS_standard_enumeration} can be improved, using a different approach. This approach is based on the following construction from \cite{KoscheKMS21}. 
More precisely, based on Theorem \ref{thm:mas}, the authors of  \cite{KoscheKMS21} define a labelled directed acyclic graph ${\mathcal D}_w$ with the nodes $\{(i,j)\mid 0\leq j< i\leq n\}\cup \{s,f\}$ (the notation $(0,0)$ is used instead of $s$ in \cite{KoscheKMS21}). The edges (represented as arrows $A \rightarrow B$ between nodes $A$ and $B$) are defined as follows:
\begin{itemize}
\item We have an edge $s\rightarrow (i,0)$, if there exists $a\in \Sigma$ such that $i=\nextpos[a,1]$ is the position where $a$ occurs the first time in $w$. This edge is labelled with $w[i]=a$.
\item For $1\leq j< i< k\leq n$ we have an edge $(i,j)\rightarrow (k,i)$,  if there exists $a\in \Sigma$ such that $k=\nextpos[a,i+1]$ is the position of the first occurrence of $a$, strictly to the right of position $i$ and $a\in \al(w[j+1:i])$. This edge is labelled with $w[k]=a$.
\item We have an edge $(i,j)\rightarrow f$ for each $b\in \al(w[j+1:i])$ such that $b\notin  \al (w[i+1:n])$. This edge is labelled with $b$. (Note that there can be multiple edges from $(i,j)$ to $f$, but they have different labels. For simplicity of exposure, we still refer to ${\mathcal D}_w$ as a DAG, as these multiple edges do not influence our reasoning at all.)
\end{itemize}
According to \cite{KoscheKMS21}, there is a bijection between $\mas(w)$ and the set of $s-f$-paths in the directed graph ${\mathcal D}_w$: the $\mas$ corresponding to an $s-f$-path is obtained by reading the label of that path, and, for each $\mas$, there is exactly one $s-f$-path whose label is the respective $\mas$. So, enumerating all elements of $\mas(w)$ is equivalent to enumerating all $s-f$-paths of ${\mathcal D}_w$, which can be done in a canonical way, with output-linear delay, by backtracking (i.e., using a depth-first search starting in node $s$). In \cite{KoscheKMS21} it is claimed that constructing ${\mathcal D}_w$ may take up to $\O(n^2\sigma )$ time, based on the fact that ${\mathcal D}_w$ may have $\O(n^2)$ nodes and each node may have $\sigma$ children. However, based on the proof of 
Theorem \ref{thm:MAS_skeleton}, one can note that for a fixed $i$, the set of $\{(r,i)\mid r\in [n], \exists\, j: (i,j)\rightarrow (r,i)\}$ has at most $\sigma$ elements. Therefore, there are $\O(n\sigma)$ nodes which are accessible from $s$, thus providing an improved upper bound of $\mathcal O(n\sigma^2)$ on the total size of the graph. To this end, we can now note that Theorem \ref{thm:MAS_skeleton} provides a compact representation of the accessible part of the DAG ${\mathcal D}_w$. 

But, we aim for an (incremental) enumeration algorithm for $\mas(w)$, with $\O(n)$ time preprocessing, with output-linear delay (respectively, $\O(1)$ delay), so we need to avoid constructing ${\mathcal D}_w$ or its compact representation explicitly, as both would require too much time. The following result shows that we can still achieve the same output-linear delay, but with only $\O(n)$ preprocessing, as we do not need to have ${\mathcal D}_w$ explicitly. Instead, while using the same backtracking/depth-first search strategy, we can construct its paths one by one, on the fly.\looseness=-1

\begin{restatable}{theorem}{masenumlindelaytheorem}\label{thm:MAS_enumeration_linear_delay}
    Given a string $w\in \Sigma^n$, over an alphabet of size $\sigma$, after an $\O(n)$ time preprocessing, we can enumerate $\mas(w)$ with output-linear delay.
\end{restatable}
\begin{proof}
\noindent\emph{$\S $ Our Approach.} 
Let us assume that we want to enumerate the elements of $\mas(w)$ starting with letter $a$, and let $i_1$ be the position of the first occurrence of $a$ in $w$. In our enumeration algorithm, we will simulate a depth-first search (DFS) of the graph ${\mathcal D}_w$ starting with node $(i_1,0)$, without having this graph explicitly. In such a DFS, the $(i_1,0)-f$-paths are explored one by one, in the following way: a path $(i_1,0)\rightarrow\ldots\rightarrow(i_r,i_{r-1})$ is stored on a stack at some point (with the node $(i_r,i_{r-1})$ found on the top of the stack), and, for each child ($(i_{r+1},i_r)$ or $f$) of $(i_r,i_{r-1})$, we put this child on the stack and explore, one by one, in a recursive way, all $s-f$-paths starting with  $s\rightarrow (i_1,0)\rightarrow\ldots\rightarrow(i_r,i_{r-1})\rightarrow (i_{r+1},i_r)$.
To implement this strategy efficiently, we take advantage of the fact that the children of some node $(i,j)$, accessible from $s$, are determined by the positions $i$ and $j$ only. More precisely, let $\al(w[j+1:i])=\{a_1,\ldots,a_r\}$, for some $r\leq \sigma$. One can compute the list of $r$ children of $(i,j)$ as follows: for $t\in [r]$, $(j_t,i)$ is a child of $(i,j)$, where $j_t$ is the  position of the first occurrence of $a_t$ strictly to the right of $i$, if $a_t\in \al(w[i+1:n])$; $f$ can also be a child of $(i,j)$ if $\al(w[j+1:i])\setminus \al(w[i+1:n])\neq \emptyset$, but, for the sake of an uniform presentation, we say, in this case, that $(n+1,i)$ (which stands in for $f$) is a child of $(i,j)$. The main challenges, that we need to address in our algorithm, are how to compute these letters $a_t\in \al(w[j+1:i])$ to efficiently gain in this way on-the-fly access to the children of $(i,j)$, and to integrate these in our simulation of the DFS of ${\mathcal D}_w$.

To this end, we will use the range maximum query data structure \cite{BenderF00}. We can preprocess in $\O(n)$ time an array $A$ with $n$ elements, so that the answer to any range maximum query is retrieved in $\O(1)$ time. In a range maximum query, denoted $RMaxQ(i,j)$, for some $i,j\in [n]$, we ask for a position $x\in[i:j]$ such that $A[x]\geq A[\ell],$ for all $\ell\in [i:j]$. In the case when there are multiple positions where the maximum of the subarray $A[i:j]$ occurs, we just return the smallest such position; in the case $i>j$, $RMaxQ(i,j)=-\infty$. We build data structures allowing us to do range maximum queries on the array $\nextArray$.

The enumeration idea implemented in our algorithm is now relatively simple to explain. Assume that we have constructed the prefix $s\rightarrow (i_1,0)\rightarrow \ldots\rightarrow (i_r,i_{r-1})$ of some $s-f$-path (and stored its nodes on a stack), which corresponds to the canonical embedding $i_1<\ldots <i_r$ of an $\mas$-prefix. We can then identify the possible positions $j$ allowing us to extend this path (respectively, $\mas$-prefix) with a node $(j,i_r)$ by selecting first the position $j=\nextArray[i_r]$, and then, one by one, the positions $j=\nextArray[j']$ for all $j'\in [i_{r-1}+1:i_r-1]$ such that $\nextArray[j']>i_r$. Indeed, this approach simply selects, one by one, each letter $w[j']\in \al(w[i_{r-1}+1:i_r])$, occurring on some position $j'$, whose next occurrence $j=\nextArray[j']$ is to the right of $i_r$ (i.e., $j>i_r$); the nodes $(j,i_r)$ computed based on these selections are exactly the children of $(i_r,i_{r-1})$, according to the observation above. Now, the first such letter $w[j']$ that we select is $w[i_r]$, as $\nextArray[i_r]>i_r$ clearly holds; this corresponds to extending the path with the node $(\nextArray[i_r],i_r)$. This node is added to the stack, and the process is repeated. Once all paths starting with $s\rightarrow (i_1,0)\rightarrow \ldots\rightarrow (i_r,i_{r-1})\rightarrow (i_{r+1},i_r )$ are enumerated (or the $\mas$ corresponding to $s\rightarrow (i_1,0)\rightarrow \ldots\rightarrow (i_r,i_{r-1})\rightarrow (n+1,i_r)$, where the last letter of the respective $\mas$ is $w[j']$, is output), we need to find a new letter $w[j']$, as described above. To this end, the main observation is that such letters $w[j']\neq w[i_r]$, and their corresponding position $j'$, can be found efficiently by using successive range maximum queries on the array $\nextArray$, on subintervals of $[i_{r-1}+1:i_r-1]$, that cover this range completely. \looseness=-1

This $RMaxQ$-based search is implemented as follows: we consider $j'_0=RMaxQ(i_{r-1}+1,i_r-1)$ as the first candidate for $j'$ and, if $\nextArray[j'_0]>i_r$, we recursively enumerate the paths starting with  $s\rightarrow (i_1,0)\rightarrow \ldots\rightarrow (i_r,i_{r-1})\rightarrow (\nextArray[j'_0],i_r)$; then we recursively look for future candidates for $j'$ in $[i_{r-1}+1:j'_0-1]$ and $[j'_0+1:i_{r}-1]$. For the sake of efficiency, we maintain for each node $(i_r,i_{r-1})$ of our stack a queue of subintervals in which we will look for future candidates for position $j'$, and we insert such an interval $[x:y]\subseteq [i_{r-1}+1:i_{r}-1]$ in the respective queue if and only if $\nextArray[RMaxQ(x,y)]> i_r$ (otherwise, that interval contains no valid candidate for the position $j'$). The first interval considered to be inserted in that queue is $[i_{r-1}+1:i_{r}-1]$, and when we want to compute a new child of $(i_r,i_{r-1})$ we simply extract the first interval stored in the queue, do a range maximum query on it, and then split it as explained above (with the two resulting subintervals being further inserted in the queue, in the case the maximum value stored in the corresponding ranges of $\nextArray[\cdot]$ is greater than $i_r$). 

The correctness of this approach follows from the construction of ${\mathcal D}_w$: indeed, $s\rightarrow (i_1,0)\rightarrow \ldots\rightarrow (i_r,i_{r-1})\rightarrow (\nextArray[j'],i_r)$ is the prefix of some $s-f$-path (if $j'$ is selected as above), and, by Theorem \ref{thm:mas}, $i_1<\ldots <i_r<\nextArray[j']$ is the embedding of an $\mas$-prefix, if $\nextArray[j']\leq n$. Our algorithm extends $s-f$-paths (and the corresponding $\mas$-prefixes) using this idea, until $\nextArray[j']=n+1$; in that case, an $s-f$-path $s\rightarrow(i_1,0)\rightarrow\ldots\rightarrow(i_\ell,i_{\ell-1})\rightarrow (n+1,i_\ell)$ (corresponding to an $\mas$) has been discovered and the $\mas$ is output (in linear time w.r.t. $\ell$). In parallel with the output of this $\mas$, we go back (right to left) along the path $s\rightarrow(i_1,0)\rightarrow\ldots\rightarrow(i_\ell,i_{\ell-1})\rightarrow (n+1,i_\ell)$, until we find a node $(i_r,i_{r-1})$ and the corresponding range $[i_{r-1}+1:i_r]$, where not all letters $w[j']$, whose next occurrence $\nextArray[j']$ is to the right of $i_r$, and the corresponding nodes $(\nextArray[j'],i_r)$ were explored as continuations of the prefix $s\rightarrow(i_1,0)\rightarrow\ldots\rightarrow(i_r,i_{r-1})$ towards obtaining an $s-f$-path. The algorithm then constructs a new $s-f$-path (and the corresponding $\mas$) by extending the path $s\rightarrow(i_1,0)\rightarrow\ldots\rightarrow(i_r,i_{r-1})$ with $(\nextArray[j'],i_r)$, and then continuing as described above. Due to the repeated use of $RMaxQ$, on subintervals of $[i_{r-1}+1:i_r]$ to identify the positions $j'$, with the properties described above, and provided that the queried subintervals completely cover $[i_{r-1}+1:i_r-1]$ (while position $i_r$ is handled separately), we can ensure that all positions $j'$ with $\nextArray[j']>i_r$ are found, and all children of $(i_r,i_{r-1})$ (and the corresponding paths) are explored.\looseness=-1

As far as the complexity of this approach is concerned, after a linear time preprocessing, in which we build the $\nextArray[\cdot]$ array for $w$ and $RMaxQ$-data structures for it, the enumeration algorithm just implements a DFS of the graph. When this DFS reaches a node $(i,j)$, we get in $\O(1)$ time the first child of that node, and then explore the paths going through that child (which extend the current $s-f$-path-prefix, maintained in the stack of the DFS). When the DFS returns to the node $(i,j)$, an unexplored child of $(i,j)$ is retrieved in $\O(1)$ time using the queue of intervals stored for that node and range maximum queries. So, while the children of a node are not explicitly stored, we can still retrieve them in order, in $\O(1)$ time per child. 


\noindent\emph{$\S $ Details.} The details of this algorithm are described in the following.

In the first step, we build the $\nextArray[\cdot]$ array for $w$, and range maximum query data structures for this array. This preprocessing step can be implemented in linear time.

Our algorithm maintains a stack ${\mathcal S}$ where we build our solution. The elements of ${\mathcal S}$ are either tuples $((n+1,j),c,\emptyset)$, where $j\in [n]$ and $c\in \Sigma$, or
tuples $((i,j),c,L)$, with $i,j\in [n]$, $c=w[i]$, and $L$ a list (implemented as queue) of pairs $[x:y]$ with $x,y\in [j+1:i]$. Initially, ${\mathcal S}$ is empty. 

We describe how our algorithm enumerates all $s-f$-paths of ${\mathcal D}_w$ which start with some node $(i_1,0)$, where $i_1$ is the first occurrence of a letter $c$ in $w$. To obtain all elements of $\mas(w)$, we simply have to iterate this approach over all letters of the alphabet (and the nodes that correspond to the first occurrences of these letters).

So, let $i_1$ be the first occurrence of $c$ in $w$. We define the list $L_1$ to contain the pair $[i_1:i_1]$, as first element, and then, as second element, the pair $[1:i_1-1]$, but only if $ \nextArray[RMaxQ(1,i_1-1)]>i_1$. We insert the element $((i_1,0),c,L_1)$ in the stack ${\mathcal S}$.

Now, while the stack ${\mathcal S}$ is not empty, we do the following. Assume that ${\mathcal S}$ contains, in order from bottom to top, the elements $((i_1,0),c_1,L_1), \ldots, ((i_r,i_{r-1}),c_r,L_r)$. Our algorithm considers two cases.

Case 1: if $i_r<n+1$, then we proceed as follows: we retrieve and remove from $L_r$ the first element $[x:y]$. Let $\ell = RMaxQ(x,y)$, $\ell_1=RMaxQ(x,\ell-1)$, and $\ell_2=RMaxQ(\ell+1,y)$. If $\nextArray[\ell_1]>i_r$, we insert $[x:\ell-1]$ in $L_r$. If $\nextArray[\ell_2]>i_r$, we insert $[\ell+1:y]$ in $L_r$. If $\nextArray[\ell]>i_r$, we define the list $L_{r+1}=\emptyset$ and we set $i_{r+1}=\nextArray[\ell]$; if, additionally, $i_{r+1}<n+1$, we insert in $L_{r+1}$ the element $[i_{r+1}:i_{r+1}]$, and then the element $[i_r+1:i_{r+1}-1]$, but only if $ \nextArray[RMaxQ(i_r+1,i_{r+1}-1)]>i_{r+1}$. Then, we insert $((i_{r+1},i_r),w[\ell],L_{r+1})$ in ${\mathcal S}$ (as the new top element), and start a new iteration of the while-loop; note that whenever $i_{r+1}\leq n$, then $w[i_{r+1}]=w[\ell]$. 

Case 2: if $i_r=n+1$, then we need to output an $\mas$, and set the stage for computing the next $\mas$. Assume that ${\mathcal S}$ contains, in order from bottom to top, the elements $((i_1,0),c_1,L_1), \ldots, ((n+1,i_{r-1}),c_r,\emptyset)$. We proceed as follows. In a first loop, we go top to bottom through ${\mathcal S}$ until we reach an element $((i_p,i_{p-1}),c_p,L_p)$ such that $L_p\neq \emptyset$; while going through the stack in this first loop, we pop all traversed elements $((i_q,i_{q-1}),c_q,L_q)$, with $L_q=\emptyset$, and construct the string $v=c_{p+1}\ldots c_r$. In a second loop, we go top to bottom through the rest of the stack ${\mathcal S}$ until we reach the bottom, without popping any element, but we extend the string $v$, by prepending to it the letters of the traversed records; that is, we construct the string $v=c_{1}\ldots c_r$. We then output this string as an $\mas$, and continue with a new iteration of the while-loop. 

When the while-loop is completed, ${\mathcal S}$ is empty, and we have produced all $\mas$s that start with letter $c\in \Sigma$. We repeat the procedure for a different letter.

\medskip

\noindent\emph{$\S $ Correctness and Complexity.}
From the description above we see that the preprocessing phase runs in linear time $\O(n)$, and the time elapsed between the output of two strings $v$ and $u$ is linear in $|u|$ (that is, the delay is output-linear). We just need to show that the algorithm correctly outputs the $\mas$s of $w$. As the described procedure enumerates all the $s-f$-paths in ${\mathcal D}_w$ in the order in which they are found in a depth-first search of the respective DAG, this is immediate.
In conclusion, we have presented an algorithm which has linear preprocessing, enumerates all $s-f$-paths in ${\mathcal D}_w$, and outputs the corresponding $\mas$s with output-linear delay. 
\end{proof}

We can extend this approach to an optimal incremental enumeration algorithm. First, we define the default edges for the nodes of the DAG ${\mathcal D}_w$, following the approach used in the previous proof: if $(i,j)$  is a node considered in our enumeration algorithm, then its first explored child is $(\nextArray[i],i)$ (or $f$, if $\nextArray[i]>n$). So, we define the default edge for the node $(i,j)$ to be the edge connecting it to the node $(\nextArray[i],i)$ (or the edge labelled with $w[i]$ connecting $(i,j)$ to $f$, if $\nextArray[i]>n$); note that, in this setting, the longest default path starting in a node $(i,j)$ ends with $f$. In this framework, we can state our result. \looseness=-1

\begin{restatable}{theorem}{masenumconstantdelaytheorem}\label{thm:MAS_enumeration_constant_delay}
    Given a string $w\in \Sigma^n$, over an alphabet of size $\sigma$, after an $\O(n)$ time preprocessing, we can incrementally enumerate $\mas(w)$ with $\O(1)$ delay.
\end{restatable}
\begin{proof}
\noindent\emph{$\S $ General Approach.}
We recall that there is a bijection between $\mas(w)$ and the labelled paths in $\mathcal D_w$. Thus, we can simply incrementally enumerate these paths, by outputting encodings (or, in other words, succinct descriptions) of these paths. So, we efficiently implement the approach from the previous theorem: we do not construct ${\mathcal D}_w$, but construct and explore its paths on the fly, while simulating a DFS of the DAG. The only difference is that we work with succinct representations of these paths (based on default edges and default paths), allowing us to move from one path to the next one, in the enumeration, in $\O(1)$ time. Other than that, the paths will be output in exactly the same order as in the previous algorithm, so the correctness will be immediate. As in the respective proof, we construct the array $\nextArray$ and  build data structures allowing us to do range maximum queries on it.

We start with some simple remarks. We have already noted that a node $(i,j)$ of ${\mathcal D}_w$ has the (default) child $(\nextArray[i], i)$ (or $f$, if $\nextArray[i]>n$). Moreover, if $i>j+1$, then the node $(i,j)$ also has the node $(\nextArray[i-1], i)$ as child (or $f$, if $\nextArray[i-1]>n$); indeed, this holds due to Theorem \ref{thm:mas} and because $w[i-1]\neq w[i]$ (as $w[i]\notin \al(w[j+1:i-1])$). So, we can check in $\O(1)$ whether a node $(i,j)$ has other children than the default child: we simply need to check whether $i>j+1$. Further, if we have a default path in ${\mathcal D}_w$ containing, in order, the nodes $(i_1,i_0),\ldots, (i_r,i_{r-1})$, we can immediately identify the rightmost node on this default path that has more than one child (i.e., another child than the target of the default edge) by finding the largest value $\ell\in [r]$ such that $i_\ell - i_{\ell-1}>1$.  \looseness=-1

In our enumeration, we will store on the DFS-stack {\em path-objects}. These are either tuples $[(n+1,j),a,\emptyset]$ (acting, like in the proof of \cref{thm:MAS_enumeration_linear_delay}, as a marker for the fact that the current path has reached $f$) or tuples $[\defPath((i,j),(i',j')), w[i], I]$, where $(i,j)$ is a node of ${\mathcal D}_w$, $(i',j')$ is a node accessible from $(i,j)$ on the longest default path leaving $(i,j)$, and $I$ is a pair $\langle (i'',j'') | L \rangle$ where $(i'',j'')$ is a node on the path encoded by $\defPath((i,j),(i',j'))$, with $i''>j''+1$, and $L$ is a queue containing intervals $[x:y]$, with $x,y\in [j''+1:i''-1]$. Each path-object $[\defPath((i,j),(i',j')), a, I]$
represents a sequence of nodes $(i_1,i_0),(i_2,i_1),\ldots, (i_r,i_{r-1})$, where $(i_1,i_0)=(i,j)$ and $(i_r,i_{r-1})=(i',j')$, and $w[i_\ell]=a$ for each $\ell\in[r]$.
Note that we might also have that $(i,j)=(i',j')$, and then $[\defPath((i,j),(i',j')), a, I]$ represents only the node $(i,j)$. Also, the path-object $[(n+1,j),a,\emptyset]$ corresponds to the pair formed by the node $f$ and a generic $a$-labelled edge of ${\mathcal D}_w$ from some node $(j,x)$ to $f$, where $x$ is not fixed; for simplicity of exposure, we will say, in the following, that the node associated with this path-object is $(n+1,j)$ (as stand in for $f$). 

As before, we focus on enumerating the paths starting with some node $(i_1,0)$, where $w[i_1]$
is the first occurrence of letter $a=w[i_1]$ in $w$; let $i_{r-1}$ and $i_r$ be the positions of the rightmost two occurrences of $a$ (where $i_r>i_{r-1}$). We can safely assume that $a$ occurs at least twice in $w$, otherwise we simply enumerate the path corresponding to the $\mas$ $aa$ and move on to another letter. The first step in our enumeration is to output ``$s,\defPath((i_1,0),(i_r,i_{r-1})),((n+1,i_r),a)$''; this corresponds to the path starting in $s$, going to $(i_1,0)$, then following the default path from $(i_1,0)$ to $(i_r,i_{r-1})$, and taking the $a$-labelled edge to $f$, and we note that this is also the first path enumerated by the algorithm of  \cref{thm:MAS_enumeration_linear_delay}. We also put on the stack maintained in our algorithm, in order, the path-objects $[\defPath((i_1,0),(i_r,i_{r-1})),w[i_1],I],[(n+1,i_r),a,\emptyset]$, with $I=\langle (i_\ell,i_{\ell-1})\mid [i_{\ell-1}+1:i_\ell-1]\rangle$ where $(i_\ell,i_{\ell-1})$ is the rightmost node on the default path $\defPath((i_1,0),(i_r,i_{r-1}))$ which has more than one child. \looseness=-1

We maintain the following invariant properties. After the first $e$ outputs, we have:
\begin{itemize}
\item The sequence of nodes defined by the path-objects stored on the stack is $(i_1,0),$ $(i_2,i_1),\ldots,$ $(i_r, i_{r-1}), (n+1,i_r)$, where there is always an edge between $(i_t,i_{t-1})$ and $(i_{t+1},i_t)$, for $t\in [1:\ell]$, under the assumption that $i_{0}=0$ and $i_{r+1}=n+1$.
\item Assume $[\defPath((i,j),(i',j')), w[i], I]$ is an object on the stack, representing the sequence of nodes  $(i_g,i_{g-1}),\ldots, (i_h, i_{h-1})$, and let $(i_d,i_{d-1})$ be the rightmost node on this path which has more than one child (i.e., $i_d>i_{d-1}+1$). Then $I=\langle (i_d,i_{d-1})\mid L\rangle $
and $L$ is a list of subintervals of $[i_{d-1}+1:i_d-1]$, where each such subinterval $[x:y]$ contains at least one position $j$ such that no path starting with $s,(i_1,0),\ldots, (i_d, i_{d-1}), (\nextArray[j],i_d)$ was enumerated yet; note that $\nextArray[j]>i_d$ holds, accordingly. Moreover, each position $j\in [i_{d-1}+1:i_d-1]$, such that no path extending $s,(i_1,0),\ldots, (i_d, i_{d-1}), (\nextArray[j],i_d)$ was enumerated yet, is contained in exactly one subinterval of $L$.
\item The first $e$ outputs correspond to the first $e$ outputs of the algorithm from \cref{thm:MAS_enumeration_linear_delay}.
\end{itemize}
These properties clearly hold after the first enumeration step.

Let us now explain how an 
enumeration step is done. We retrieve from the stack the uppermost element $[\defPath((i,j),(i',j')), w[i], I]$ with $I\neq \emptyset$.
This element becomes the new top of the stack (we implement this using a second stack, as in Theorem \ref{thm:DAG_enumeration}). Assume now that $I=\langle (i'',j'') | L \rangle$ and extract $[x:y]$, the first element of $L$. We first compute $g=RMaxQ(x,y)$
(and, by the maintained properties $\nextArray[g]>i''$); if $\nextArray[RMaxQ(x,g-1)]>i''$, add $[x:g-1]$ to $L$ and if $\nextArray[RMaxQ(g+1,y)]>i''$, add $[g+1:y]$ to $L$. Note that the $RMaxQ$-based search strategy, and the corresponding interval splitting, that we employ here is the same as in the algorithm described in  \cref{thm:MAS_enumeration_linear_delay}. It is important in our enumeration that if $L$ becomes empty in this step, then we efficiently identify the rightmost node $(z,v)$ on $\defPath((i,j),(i'',j''))$ for which $z>v+1$ and set $I'=\langle (z,v)\mid [v+1:z-1]\rangle$; if no such node exists, we set $I'=\emptyset$. 
Now, let us assume that $\nextArray[g]\leq n$. Let $\defPath((\nextArray[g],i''),(g',g''))$ be the longest default path starting in node $(\nextArray[g],i'')$ and with $g'\leq n$.
We can now emit the new output: ``return to node $(i'',j'')$, follow the edge to node $(\nextArray[g],i'')$, follow $\defPath((\nextArray[g],i''),(g',g''))$, follow edge to $f$ labelled with $w[g]$''. Then we pop $[\defPath((i,j),(i',j')), w[i], I]$ from the top of the stack and insert, in this order, the path-objects $[\defPath((i,j),(i'',j'')), w[i], I']$, $[\defPath((\nextArray[g],i''),(g',g'')), w[g], I'']$, 
and $[(n+1,g'),w[g],\emptyset]$. Here, $I''=\langle (p,q)\mid [p+1:q-1]\rangle$, where $(p,q)$ is the last node on $\defPath((\nextArray[g],i''),(g',g''))$ with $p>q+1$ (i.e., $(p,q)$ has more than one child). 
If $\nextArray[g]= n+1$, we can directly emit the new output: ``return to node $(i'',j'')$, follow edge to $f$ labelled with $w[g]$''. Then we pop $[\defPath((i,j),(i',j')), w[i], I]$ from the stack, insert $[\defPath((i,j),(i'',j'')), w[i], I']$ and also insert $[(n+1,i''),w[g],\emptyset]$. Then we move on and enumerate the next path.

As we basically identify the rightmost node which has an unexplored child on the currently enumerated path, it is immediate that this algorithm simulates exactly the one from \cref{thm:MAS_enumeration_linear_delay}, so the last item of our invariant is preserved. The first and second item are also preserved: While the first one is trivial, the argument for the second one is more involved. Firstly, we always work on the interval list corresponding to the last node with unexplored children from a default path, and we make sure that this always holds. Then, the interval queues are rewritten using range maximum queries on $\nextArray$: we select one interval, identify a position in it corresponding to a new child, and then we split the interval in two subintervals around this position (which is now explored). This is analogous to the procedure in 
\cref{thm:MAS_enumeration_linear_delay}, and it ensures that we will explore all children of that node (we do not miss any position leading to an unexplored child). As all these invariants hold, the correctness of our algorithm also follows.

For the complexity (see also below), note that the delay between two outputs is constant: the edit-scripts are of constant size, and the time needed to move from one path to the next, while correctly managing the stack, is also $\O(1)$ due to the usage of the data structures described next. Further, all these data structures are computed in $\mathcal O(n)$ time.

\noindent\emph{$\S $ Details.} Our algorithm maintains a stack (corresponding to the DFS) which is implemented statically: it is an array (of size at most $n+1$) for which we maintain a variable $top$ that points to the top of the stack (incrementing it at push operations, decrementing it at pop operations), and the objects of the stack are stored in the respective array (the bottom is on position $1$, the top on position $top$). This allows us to pop in $\O(1)$ time multiple elements of the array: to pop $\ell$ elements, we simply decrement $top$ by $\ell$; in subsequent pushes, the elements which were previously in the array are overwritten. 

In our algorithm, we need to be able to retrieve the uppermost element of the stack which has $I\neq \emptyset$. This is done by maintaining an additional auxiliary stack of the elements with $I\neq \emptyset$, with pointers to their position in the DFS-stack. When saying  that we retrieve from the DFS-stack the uppermost element $[\defPath((i,j),(i',j')), w[i], I]$ with $I\neq \emptyset$ and this element becomes the new top of the stack, we simply update the $top$ variable of the DFS-stack to point to the element $[\defPath((i,j),(i',j')), w[i], I]$, which was on top of the auxiliary stack. Modifications on the elements of the DFS-stack, as well as pops and pushes in this stack, should be reflected in the auxiliary stack, but this is trivial, as we always operate only on the top elements of both stacks.

In $\O(n)$ overall time (by traversing the word once), we can produce $\sigma$ lists $L_a$, for $a\in\Sigma$, containing the positions where $a$ occurs in $w$, increasingly ordered. 
We can assume that each list has at least two elements: if $a$ occurs only once in $w$, then the last element of $L_a$ is $n+1$. The last two elements of $L_a$ will give us then, in $\O(1)$ time, the last two occurrences of $a$ in $w$, as needed in our algorithm, for the initial enumeration step. We store these two occurrences (say $x$ and $y$, with $x<y$) as $\lastPair[a]=(x,y)$. 

We will also use the lists $L_a$, for $a\in \Sigma$, to define an array $\lastGap[\cdot]$, where $\lastGap[i]=[x:y]$ if and only if $\al(w[y:i])=\{w[i]\}$, $w[x]=w[y]=w[i]$, $y=\nextArray[x]$ and $x-y>1$, and  $\lastGap[i]=\emptyset$ otherwise. For each $a\in \Sigma$, we traverse $L_a$ from left to right. If $i$ is the first element of $L_a$, then $\lastGap[i]=\emptyset$. If $i$ is not the first element of $L_a$, and $j$ is the element occurring just before $i$ in $L_a$ (i.e., $j<i$ and $a\notin \al(w[j+1:i-1])$), then $\lastGap[i]=[j:i]$, if $i-j>1$, or $\lastGap[i]=\lastGap[j]$, otherwise. We repeat this process for all $a\in \Sigma$, and the total time is given by the total size of the lists $L_a$, for $a\in \Sigma$, so $\O(n)$. 
We use this array to identify the rightmost node $(z,v)$ on a default path (represented succinctly as) $\defPath((i,j),(i'',j''))$ for which $z>v+1$. We proceed as follows. We retrieve $\lastGap[i'']=[x:y]$. If $x\geq j$, then $(y,x)$ is the node $(z,v)$ we were looking for (by the definition of the default path, this node must be on that path). If $x<j$, then $(i,j)$ is the node $(z,v)$ we were looking for, but only if $i\neq j+1$. Finally, in all other cases, there is no such node $(z,v)$. 

Finally, we need to compute in $\O(1)$ time the node $(g',g'')$ with $g'\leq n$ such that $\defPath((\nextArray[g],i''),(g',g''))$ is the longest default path starting with $(\nextArray[g],i'')$.
We retrieve $\lastPair[w[g]]=(x,y)$, and note that $y\geq \nextArray[g]$ (as, in the worst case, both $y$ and $\nextArray[g]$ equal $n+1$).
If $x\geq \nextArray[g]$, then $g''=x$ and $g'=y$.  If $x < \nextArray[g]=y$, then $g''=y$ and $g'=n+1$. 

Using these data structures (which can be computed in $\O(n)$ time) alongside the array $\nextArray$ and  the corresponding data structures allowing us to do range maximum queries on it (which can also be computed in $\O(n)$ time), we are able to move from the currently enumerated path (and the corresponding output) to the next one in $\O(1)$ time.

\noindent\emph{$\S $ Correctness and Complexity.} The algorithm simply simulates, more efficiently, the one in \cref{thm:MAS_enumeration_linear_delay}: its outputs correspond one-to-one to the outputs of that algorithm. Therefore, the correctness is a consequence of the correctness of that algorithm (and, in particular, it follows from the fact that both these algorithms just enumerate the paths of ${\mathcal D}_w$ using a DFS). As far as the complexity is concerned, the delay between two outputs is constant:
the edit-scripts are of constant size, and the time needed to move from one path to the next, while correctly managing the stack, is also bounded by a constant, as explained in the part dedicated to the used data structures. All these data structures are computed in $\O(n)$ time, so the preprocessing stays linear. 

Moreover, it is immediate that the currently enumerated path (that is, $\mas$) can be retrieved from the stack of the DFS in output-linear time; for the paths, we simply have to output the nodes, in the order in which they appear on the default paths stored on the respective stack, from the bottom to the top. The $\mas$ can then be derived from this sequence of nodes (with its final letter being the second component of the top symbol $[(n+1,i''),w[g],\emptyset]$ of the stack). Thus, our statement holds.
\end{proof}

According to Section \ref{sec:enumSkelDAG}, skeleton-representations for $\sas(w)$ or $\mas(w)$ can also be used to count their elements. It is unclear if the approach of Theorem \ref{thm:MAS_enumeration_constant_delay} can be used in this direction. \looseness=-1

\section{Conclusion and Further Work}\label{sec:longMAS}

In this paper, we have given optimal (incremental) enumeration algorithms for the sets of shortest and minimal absent subsequences of a given input word. A potential way to improve these results would be to aim for algorithms with the same optimal preprocessing time and delay which enumerate the respective sets in some natural order of their elements (e.g., in lexicographical order). Our algorithms do not have such a nice property, as the order in which the respective absent subsequences are enumerated seems somehow arbitrary.

Another interesting direction of research would be to understand the extremal elements of these sets, with respect to some natural order. For instance, when we consider ordering these sets with respect to the length of the absent subsequence, then, for a given word $w$, all elements of $\sas(w)$ have the same length, the shortest elements of $\mas(w)$ have the same length as the elements of $\sas(w)$, but computing a longest element of $\mas(w)$ seems non-trivial. In particular, an algorithm computing this length, for a string $w$ of length $n$ over an alphabet of size $\sigma$, in $\O(n\sigma)$ time was reported in \cite{Tronicek23}. We can improve this bound significantly, and propose here an algorithm computing a longest $\mas$ in $\O(n\log \sigma)$ time. 

Let us begin the presentation of this algorithm by emphasising that simple greedy strategies for computing a longest $\mas$ of a string do not seem to work. For instance, one possible greedy strategy 
(suggested, e.g., by looking at a longest $\mas$ of the string $1^2 2 1^5$, from Section \ref{sec:prels}, namely $1^8$), 
would be to compute a longest $\mas$ in a string $w$ by simply selecting the letter $c$ which has the most occurrences in $w$, and, if its number of occurrences is $\ell$, then return $c^{\ell+1}$ as a longest $\mas$. Clearly, $c^{\ell+1}$ is an $\mas$ of $w$, but not necessarily a longest one. Indeed, if $w=1^{n}212^{n}$, then a longest $\mas$ of $w$ is $1^{n+1} 2^{n+1}$, while the above mentioned greedy strategy would return $1^{n+2}$ or $2^{n+2}$. So, it seems that a more sophisticated approach is needed, and we propose a dynamic programming algorithm solving this problem.


In order to define our algorithm, we need the following data structure, which can be, e.g., a 1-dimensional range tree~\cite{Lueker78,Bentley79}, implemented using AVL-balanced trees \cite{AVL}. 
\begin{restatable}{lemma}{avlmaxlemma}\label{lem:AVL-MAX}
Let $n,\sigma$ be positive integers, with $\sigma\leq n$. We can maintain a set $S\subset [n]\times[n] $ with $\sigma$ elements, where if $(a,b),(c,d)\in S$ then $a\neq c$, such that the following operations are implemented in $\O(\log \sigma)$ time:
\begin{itemize}
\item $\mathrm{insert}(a,b)$: inserts $(a,b)$ in $S$; it is only successful if there is no element $(a,d)$ already contained in $S$ and $|S|<\sigma$. 
\item $\mathrm{delete}(a)$: removes $(a,b)$ from $S$ (the only element with $a$ as its first component), if such an element exists.
\item $\mathrm{rangeMax}(x,y)$: returns $(c,d)\in S$ with $c\in [x:y]$ and $d=\max\{b\mid (a,b)\in S, a\in [x:y]\}$ and, for any other $(e,d)\in S$, with $e\in [x:y]$, we have $c<e$.
\end{itemize}
\end{restatable}

\begin{proof}
The general idea is to maintain $S$ as an AVL-balanced tree \cite{AVL}, where the order in the tree is according to the first component of the pairs contained in $S$. Moreover, each node of the tree will store the maximum of the second components of the elements stored in its subtree. This structure is usually called a 1-dimensional range tree \cite{Lueker78,Bentley79}. Insertions and deletions are done as in a usual AVL-tree (w.r.t. the first component of the inserted or deleted element), with the maximum-of-subtree value being updated bottom up (also while rebalancing). Answering $\mathrm{rangeMax}(x,y)$ is done by computing the successor of $x$ and predecessor of $y$ (w.r.t. the first component), and then returning the maximum-of-subtree value stored in their lowest common ancestor.


Let us now describe everything in more details.

As mentioned above, insertions and deletions are done as in a usual AVL-tree (where the search criterion is the first component of the inserted or deleted element), the only observation being that the maximum of each subtree (w.r.t. the second component) is updated bottom-up from the inserted node (or, respectively, the deleted node) to the root of the tree; one has to take into account the potential rebalance-operations (rotations), but only $\O(\log \sigma)$ such operations are used in any insertion and/or deletion, and the maximum-of-subtree value can be recomputed in each rotation for each of the (at most) three subtrees involved.

The answer to $\mathrm{rangeMax}(x,y)$ is computed in a straightforward way: find the pair $(x',v_1)$ with a minimal $x'\geq x$ (which is equivalent to a successor search for the key $x$); find the pair $(y',v_2)$ with a maximal $y'\leq y$ (which is equivalent to a predecessor search for the key $y$); find the lowest common ancestor of $(x',v_1)$ and $(y',v_2)$, and return the maximum-of-subtree value stored in that node. The procedure is clearly correct and takes $\O(\log \sigma)$ time.

We note that the time bounds reported in Lemma \ref{lem:AVL-MAX} can be slightly improved (see \cite{ChanT18}), to $\O(\log \sigma / \log\log \sigma)$. However, for the sake of clarity and simplicity, we prefer to use the  slightly less efficient version presented above (but note that this lemma is used as a black-box, with a different complexity being reflected only in a potential replacement, by $\log \sigma / \log\log \sigma$, of the $\log \sigma$ factor present in the complexity of the algorithm reported in Theorem \ref{thm:longestMAS}). 
\end{proof}

Our result regarding the computation of the longest $\mas$ of a word is the following.

\begin{restatable}{theorem}{longestmastheorem}\label{thm:longestMAS}
Given a string $w\in \Sigma^n$, over an alphabet of size $\sigma$, we can compute a longest $\mas$ of $w$ in $\O(n \log \sigma)$ time.
\end{restatable}
\begin{proof}
\noindent {\em $\S $ General Approach.} We will obtain a longest $\mas$ by a dynamic programming approach, in which two arrays $D[\cdot]$ and $\back[\cdot]$, each with $n$ elements and indexed by the numbers of $[n]$, are computed. 

More precisely, for $i\in [n]$, we define $D[i]$ as the length of a longest $\mas$-prefix, whose canonical embedding in $w$ ends on position $i$ (note that there always exists at least one $\mas$-prefix ending on each position of $w$ by Lemma~\ref{lem:MASprefix}).

If there is a single longest $\mas$-prefix of length $D[i]=r$, whose canonical embedding in $w$ ends on position $i$, and this $\mas$-prefix has the canonical embedding $i_1<i_2<\ldots <i_{r-1}<i_r=i$, then we define $\back[i]=i_{r-1}$ (in the case when $r=1$, we set $\back[i]=0$). If there are more such longest $\mas$-prefixes of length $D[i]=r$, whose canonical embedding in $w$ ends on position $i$, we choose the one such $\mas$-prefix which has the canonical embedding $i_1<i_2<\ldots <i_{r-1}<i_r=i$ with the value of $i_{r-1}$ minimal (compared to all other such $\mas$-prefixes), and set $\back[i]=i_{r-1}$. 

\medskip

\noindent {\em $ \S$ Preliminaries.} Before showing how to compute $D[\cdot]$ and $\back[\cdot]$, let us note the following property (crucial for our algorithm): if $D[i]=r$ and $\back[i]=j$, for some $j>0$, then $D[j]=r-1$. This can be shown as follows. Clearly, by definition, $D[j]\geq r-1$ and $\prevArray[j]< \prevArray[i]\leq j$, due to Theorem~\ref{thm:mas}. Assume, for the sake of a contradiction, that $D[j]>r-1$. Then, there exists an $\mas$-prefix of length $p>r-1$ whose canonical embedding is $j_1<j_2<\ldots< j_{p-1} <j_p=j$; for simplicity, set $j_0=0$. By the definition of an $\mas$-prefix, we have that $j_{p-1}\geq \prevArray[j]>j_{p-2}$. If $j_{p-1}< \prevArray[i] $, then $j_1<j_2<\ldots< j_{p-1} <j_p=j<i$ is the (canonical) embedding of an $\mas$-prefix of length $p+1 > r$ ending on position $i$, a contradiction. So, assume $j_{p-1}\geq \prevArray[i] $. We thus have  $j_{p-1}> \prevArray[i] > \prevArray[j]>j_{p-2}$. So, $j_1<j_2<\ldots< j_{p-2}<j_{p-1}<i$ is a canonical embedding of an $\mas$-prefix of length $p\geq r$, whose one-before-last position $j_{p-1}$ is strictly smaller than $\back[i]$; again, this is a contradiction with the definitions of $D[\cdot]$ and $\back[\cdot]$.

So, based on this property, we can intuitively see that the array $D[\cdot]$ is used for computing the length of a longest $\mas$, while $\back[\cdot]$ will be used to actually retrieve this $\mas$ (from the right to the left). Indeed, if $D[i]=r$ then there exists an $\mas$-prefix of length $r$ with the embedding $0=i_0<i_1<i_2<\ldots <i_{r-1}<i_r=i$ such that $\back[i_{j}]=i_{j-1}$, for all $j\in [r]$. So, if we see the pairs $(\back[i],i)$ as directed edges in a DAG with nodes $\{0,1,\ldots, n\}$, then every $\mas$-prefix corresponds to a path starting in $0$. 

\medskip

\noindent {\em $\S $ Algorithm.} We now show how to actually compute the arrays $D[\cdot]$ and $\back[\cdot]$. Our algorithm works iteratively, for $j=0$ to $n$. In the iteration for $j=0$, the values of $D[g]$ are set, for every $g$ where a letter of $\Sigma$ occurs for the first time. Then, in the iteration for $j=i>0$ we compute $D[\nextArray[i]]$ and $\back[\nextArray[i]]$. Initially, all elements of the arrays $D[\cdot]$ and $\back[\cdot]$ are set to $-1$.

We maintain a set $E\subseteq (\{0\}\cup [n])\times (\{0\}\cup [n])$ with at most $\sigma$ elements, implemented as a $1$-dimensional range tree (see Lemma \ref{lem:AVL-MAX}); intuitively, when the iteration for $j=i$ is reached, $E$ contains the pairs $(g,v)$ (corresponding to edges $(\back[g],g)$, when one refers to the DAG-analogy)  with $-1\neq \back(g)< i\leq g$ and $D[g]=v$ (we will ensure that this is possible, since $\prevArray[g]\leq \back[g]$, $D[g]$ was already computed at this point). Initially, we set $E=\emptyset$. 

Finally, we also maintain the sets $L[q]\subseteq  (\{0\}\cup [n])\times (\{0\}\cup [n])$, for all $q\in \{0\}\cup [n]$, all initially set to $\emptyset$ and implemented as simple lists; intuitively, at the beginning of the iteration for $j=i$, for all $g$ such that $\prevArray[g]< i$, we have that the pair $(\back[g],g)$ is stored in the list $L[\back[g]]$ and there are no other pairs in the lists $L[\cdot]$. Referring to the DAG-analogy: the edges of the DAG of $\mas$-prefixes that we have discovered so far are stored in the list corresponding to their originating nodes. After a pair $(\back[g],g)$ is inserted in the list $L[\back[g]]$, during the iteration for some $j=i<\back[g]$, it will then be reconsidered (that is, inserted in $E$) when we reach the iteration $j=\back[g]$, at whose end the value $D[f]$ is computed, and $(\back[f],f)$ is inserted in $L[\back[f]]$, for $f=\nextArray[\back[g]]$.

Now we proceed with the description of the actual dynamic programming algorithm. As said above, we iterate for $j$ from $0$ to $n$. 

Before starting the computation, we set all elements of $D[\cdot]$ and $\back[\cdot]$ to~$-1$.

Initially, in the iteration for $j=0$, we simply set $D[i]=1$ and $\back [i]=0$, insert $(i,1)$ in $E$, and the pair $(0,i)$ in $L[0]$, whenever $i$ is the position of the leftmost occurrence of $w[i]$ in $w$.

We explain, in the following, how we proceed for $j\geq 1$. As already hinted above, we aim to maintain the following properties: at the end of the iteration for $j=i-1$ (that is, if $i\leq n$, at the beginning of the iteration for $j=i$, or at the end of the algorithm for $i=n+1$), the following hold:
\begin{itemize}
    \item all elements $D[\ell]$ and $\back[\ell]$ with $\ell\in [n]$ and $\prevArray[\ell]< i$ are correctly computed,
    \item $E$ contains the pairs $(g,D[g])$, where $ \back[g]< i\leq g$ and $\back[g]\neq -1$,
    \item for all $g$ such that $\prevArray[g]< i$, we have that the pair $(\back[g],g)$ is stored in the list $L[\back[g]]$ and there are no other pairs stored in the lists $L[\cdot]$.
\end{itemize}
These clearly hold at the end of the iteration for $j=0$ (i.e., beginning of the iteration for $j=1$).

Now, in the iteration for $j=i$, if $E\neq \emptyset $ and $ \nextArray[i] \leq n$, we set $h=\nextArray[i]$, $D[h]=r+1$, $\back[h]=q$, where $(q,r)=\mathrm{rangeMax}(i,h-1)$, and insert the pair $(\back[h],h)$ in $L[\back[h]]$. After that, we remove the pair $(i,t)$ from $E$ (that is, we remove from $E$ the edge ending on $i$) and for every $(i,s)\in L[i]$, we insert the pair $(s,D[s])$ in $E$.

After the iteration for $j=n$ is finished, the length of a longest $\mas$ of $w$ is $r=1+\max\{D[i]\mid i\in [n]\}$, and such a longest $\mas$ can be obtained as $w[i_1]w[i_2]\cdots w[i_{r-1}]c$, where $i_{r-1}$ is defined as a position $h$ such that $D[h]=r-1$, $i_{\ell-1}=\back[i_\ell]$ for $\ell \in [2:r-1]$, and $c$ is a letter which does not occur in $\al(w[h+1:n])$ (such a letter could be, for example, exactly $w[h]$, as another occurrence of this letter to the right of $h$ would lead to a longer $\mas$-prefix occurring in $w$).

\medskip

\noindent {\em $ \S$ Correctness and complexity.}
To show the correctness of this algorithm, we can prove by induction on $i$ that, at the end of iteration for $j=i$, the three properties that we aim to maintain (at the end of the iteration for $j=i$) do indeed hold. As said already, these properties hold at the end of the iteration for $j=0$. Assume now that the respective properties hold at the end of the iteration for $j=i-1$, and we show that they also hold at the end of the iteration for $j=i$. Let $h=\nextArray[i]$. We want to show that, if $h\leq n$, $D[h]$ and $\back[h]$ are correctly computed. Let $i_1<i_2<\ldots<i_r=h$ be the canonical embedding of a longest $\mas$-prefix ending on position $h$, and let $i_0=0$. As $\prevArray[h]=i$, then $\prevArray[i_{r-1}] < i \leq i_{r-1}<i_r$. So $D[i_{r-1}]$ and $\back[i_{r-1}]$ are correctly computed, and $(i_{r-1}, D[i_{r-1}])$ is contained in $E$ (by induction). Also, there is no position $q<h$ such that $\prevArray[q]<i$ and $D[q]>r-1$ or $D[q]=r-1$ and $q<i_{r-1}$ (this would contradict the hypothesis that $i_1<i_2<\ldots<i_{r-1}<i_r=h$ is the canonical embedding of a longest $\mas$-prefix ending on position $h$, as such an embedding should contain $q$ as the predecessor of $h$). So, $(i_{r-1},D[i_{r-1}])$ is retrieved by the $\mathrm{rangeMax}(i,h-1)$, and $D[h]$ is correctly set to $r$ and $\back[h]$ is set to $i_{r-1}$. The properties of $E$ and the lists $L[\cdot]$ are clearly correctly maintained. We also note that $|E|$ remains at most $\sigma$. At the beginning of the iteration for $j=i$, $E$ contains the pairs $(g,D[g])$, where $ \back[g]< i\leq g$ and $\back[g]\neq -1$. Assume that $E$ contains two elements $(g,D[g])$ and $(g',D[g'])$ with $w[g]=w[g']$, and assume $\prevArray[g]<\prevArray[g']$; as $\prevArray[g]<\prevArray[g']\leq \back[g']< i$ and both $g\geq i$ and $g'\geq i$, this immediately yields a contradiction. So, $E$ contains at most one pair $(g,D[g])$ such that $w[g]=a$, for every $a\in \Sigma$. So $|E|\leq \sigma$. 

It is now not hard to show that the complexity of the algorithm is $\O(n\log \sigma)$, under the assumption that the set $E$, which has at most $\sigma$ elements, is maintained as in Lemma \ref{lem:AVL-MAX}. This concludes our proof.
\end{proof}

In connection to \cref{thm:longestMAS}, we leave as an open problem the natural question whether we can improve the complexity reported in the respective theorem to $\O(n)$.

\bibliographystyle{plainurl}
\bibliography{references}

\appendix

\section{Figures and Pseudocodes}\label{sec:figures}

This section contains figures exemplifying the construction of skeleton-DAGs for $\mas(w)$ and $\sas(w)$ (Figures~\ref{fig:massas1} and \ref{fig:massas2}) as well as pseudocodes implementing the construction of these skeleton-DAGs (Algorithm~\ref{SAS:algo:main}, and Algorithms~\ref{MAS:algo:KM}~and~\ref{MAS:algo:main}, respectively). 

\begin{algorithm}[tbh]
    \caption{compute$\sas G$($w$)}\label{SAS:algo:main}
    \footnotesize
    \KwData{String $w$ with $\abs{w}=n$, alphabet $\Sigma=$ alph$(w)$ with $\abs{\Sigma}=\sigma$}
    \KwResult{Skeleton DAG $G$ consisting of $\mathcal O(n)$ nodes and edges representing all $\sas$ of $w$}
    compute arch factorisation of $w$, $\dist[\cdot]$, $\firstPosArch[\cdot][\cdot]$, and $\lastPosArch[\cdot][\cdot]$\;
    \For{$\ell=1$ \KwTo $k-1$}{
        \For{$i \in \text{ar}_w(\ell)$}{
            \lIf{$i = \firstPosArch[\ell][w[i]]$ $\wedge$ $\dist[i] = k - \ell + 2$}{add $i$ to $F_\ell$}
            \If{$i = \lastPosArch[\ell][w[i]] \wedge \dist[\firstPosArch[\ell+1][w[i]] = k - (\ell + 1) + 2$}{add $i$ to $G_\ell$\;}
        }
    }
    \For{$i \in \text{ar}_w(k)$}{
        \lIf{$i = \firstPosArch[k][w[i]]$ $\wedge$ $\dist[i] = 2$}{add $i$ to $F_k$}
        \lIf{$i = \lastPosArch[k][w[i]] \wedge w[i] \notin \text{alph}(r(w))$}{add $i$ to $G_k$}
    }
    $V \gets \{s, f\} \cup \bigcup_{\ell=1}^k F_\ell \cup \{n+i \mid i \in [\sigma]\}$\;
    \For{$i = 1$ \KwTo $\abs{F_1}$}{
        \lIf{$i+1 < \abs{F_1}$}{add $(F[i+1], F[i])$ to $E$}
    }
    \For{$\ell = 2$ \KwTo $k$}{
        \For{$i = 1$ \KwTo $\abs{G_\ell}$}{
            \lIf{$i+1 < \abs{G_\ell}$}{add $(\firstPosArch[\ell][w[G_\ell[i]]], \firstPosArch[\ell][w[G_\ell[i+1]]])$ to $E$}
        }
    }
    \For{$i = 1$ \KwTo $\abs{G_k}$}{
        \lIf{$i+1 < \abs{G_k}$}{add $(n+w[G_k[i]], n+w[G_k[i+1]])$ to $E$}
    }
    \For{$\ell = 1$ \KwTo $k$}{
        $d \gets 1$\;
        \For{$i = 1$ \KwTo $\abs{F_\ell}$}{
            \lWhile{$G_\ell[d] < F_\ell[i]$}{
                $d \gets d+1$}
            \lIf{$\ell+1 \leq k$}{
                add $(F_\ell[i], \firstPosArch[\ell+1][w[G_\ell[d-1]]])$ to $E$}
            \lIf{$\ell+1 = k+1$}{
                add $(F_\ell[i], n+w[G_\ell[d-1]])$ to $E$}
        }
    }
    
    \Return{$(V, E)$.}
\end{algorithm}

\begin{algorithm}[tbh]
    \caption{computeP($w$)} \label{MAS:algo:KM}
    \footnotesize
    \KwData{String $w$ with $\abs{w}=n$, alphabet $\Sigma=$ alph$(w)$ with $\abs{\Sigma}=\sigma$}
    \KwResult{Matrix $\KM$ such that $\KM[i,\cdot]$ contains the set $P_{i}$ for $i \in [n]$, and $\KM[i,j]$ maintains the internal total order for increasing  $j \in [\sigma]$.} 
    array \knp\  of length $\sigma$ with elements initially set to $0$\;
    array \knn\  of length $\sigma$ with elements initially set to $n+1$\;
    matrix \KM\  of size $n\times\sigma$ with elements initially set to $\infty$\;
    $last\gets 0$\;
    $first \gets n+1$\;
    \For{$i=1$ \KwTo $n$}{
        $x\gets w[i]$\;
        \If{$\knn[x]=n+1$}{
            \If{$\knp[x]=0$}{
                \If{$first=n+1 \vee w[first] = x$}{$first\gets i$\;}
                \Else{
                    $\knp[x]\gets last$\;
                    $\knn[w[last]]\gets i$\;
                }
                $last\gets i$\;
            }
            \Else{
                $\knn[w[\knp[x]]]\gets i$\;
                $last\gets i$\;
            }
        }
        \Else{
            \If{$x = w[first]$}{$first\gets \knn[x]$\;}
            \Else{
                $tmp \gets \knn[x]$\;
                $\knn[w[\knp[x]]]\gets tmp$\;
                $\knp[w[tmp]]\gets \knp[x]$\;
            }
            $\knp[x]\gets last$\;
            $\knn[x]\gets n+1$\;
            $\knn[w[last]]\gets i$\;
            $last\gets i$\;
        }

        $current\gets first$\;
        $row \gets 1$\;
        \While{$current\neq n+1$}{
            $\KM<row>[i]\gets current$\;
            $row\gets row+1$\;
            $current\gets \knn[w[current]]$\;
        }
    }
    \Return{\textup\KM.}
\end{algorithm}

\begin{algorithm}[tbh]
    \caption{compute$\mas G$($w$)}\label{MAS:algo:main}
    \footnotesize
    \KwData{String $w$ with $\abs{w}=n$, alphabet $\Sigma=$ alph$(w)$ with $\abs{\Sigma}=\sigma$}
    \KwResult{Skeleton DAG $G$ consisting of $\mathcal O(n\sigma)$ nodes and edges representing all $\mas$ of $w$}
    compute $\nextpos[\cdot,\cdot]$, and $\KM$\;
    array $S$ of length $n$ with elements initially set to $1$\;
    \For{$a=1$ \KwTo $\sigma$}{
        $i\gets\nextpos[1,a]$\;
        add $(s,v_{i}^{1})$ to the list of edges of $G$\;
    }
    \For{$i=1$ \KwTo $n$}{
        \For{$\ell=1$ \KwTo $\sigma$}{
            $j\gets \KM<\ell>[i]$\;
            \If{$j\neq \infty$}{
                $k \gets \nextpos[w[j],i+1]$\;
                \If{$k\leq n$}{
                    $s\gets S[k]$\;
                    \While{$\KM<s>[k]\leq i$}{$s\gets s+1$\;}
                    add $(v_{i}^{\ell},v_{k}^{s})$ to the list of edges of $G$\;
                    $S[k]\gets s$\;
                }
                \Else{
                    add $(v_{i}^{\ell},f)$ to the list of edges of $G$\;
                }
                \If{$\ell>1$}{
                    add $(v_{i}^{\ell-1},v_{i}^{\ell})$ to the list of edges of $G$\;
                }
            }
        }
    }
    
    \Return{$G$.}
\end{algorithm}

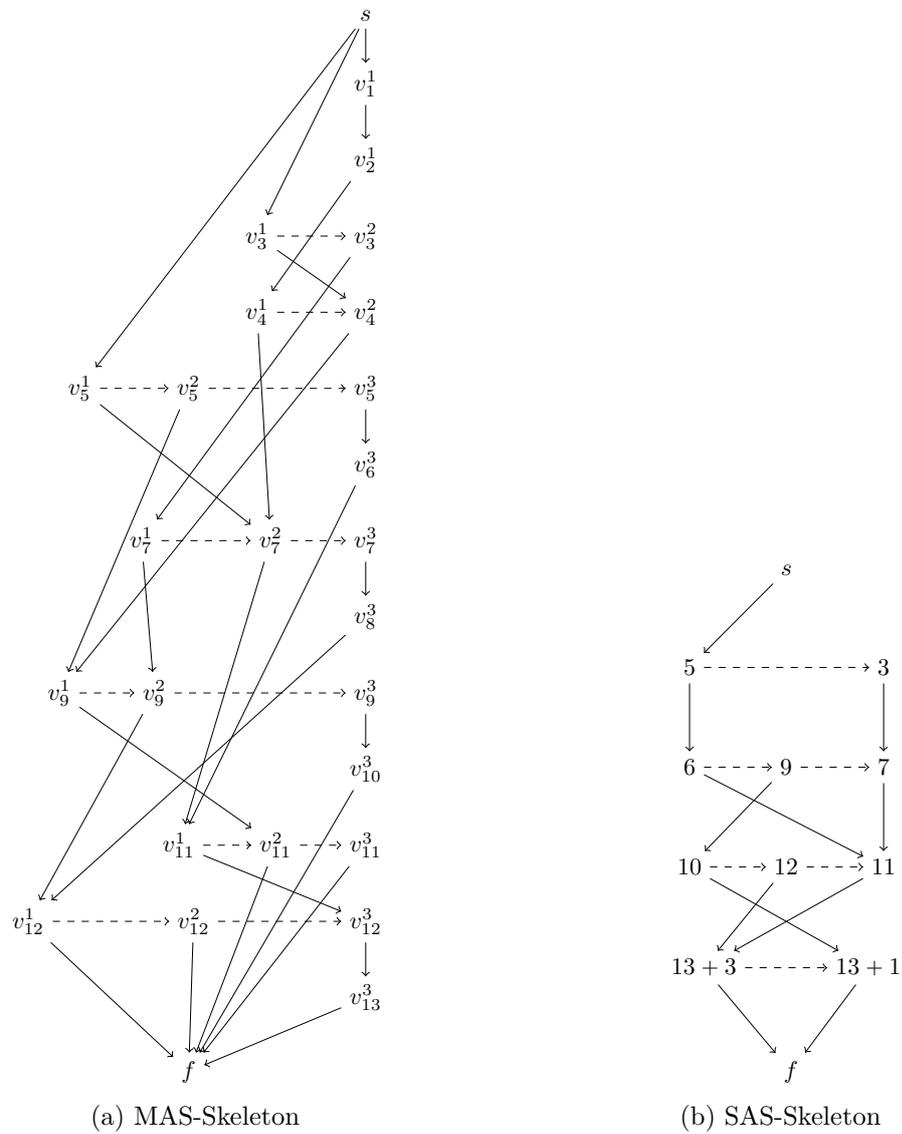
\begin{figure}[htb]
\centering
    \begin{minipage}[b]{.6\textwidth}
    \centering
        \scalebox{.9}{\begin{tikzpicture} 
            \node (0_0) {$s$};
            \node (1_0) [below=.5cm of 0_0] {$v_1^1$};
            \node (2_1) [below=.5cm of 1_0] {$v_2^1$};
            \node (3_2) [below=.5cm of 2_1] {$v_3^2$};
            \node (3_1) [left=of 3_2] {$v_3^1$};
            \node (4_3) [below=.5cm of 3_2] {$v_4^2$};
            \node (4_2) [left=of 4_3] {$v_4^1$};
            \node (5_4) [below=.5cm of 4_3] {$v_5^3$};
            \node (5_3) [left=2cm of 5_4] {$v_5^2$};
            \node (5_2) [left=1.cm of 5_3] {$v_5^1$};
            \node (6_5) [below=.5cm of 5_4] {$v_6^3$};
            \node (7_6) [below=.5cm of 6_5] {$v_7^3$};
            \node (7_5) [left=.8cm of 7_6] {$v_7^2$};
            \node (7_3) [left=1.3cm of 7_5] {$v_7^1$};
            \node (8_7) [below=.5cm of 7_6] {$v_8^3$};
            \node (9_8) [below=.5cm of 8_7] {$v_9^3$};
            \node (9_7) [left=2.5cm of 9_8] {$v_9^2$};
            \node (9_5) [left=.8cm of 9_7] {$v_9^1$};
            \node (10_9) [below=.5cm of 9_8] {$v_{10}^3$};
            \node (11_10) [below=.5cm of 10_9] {$v_{11}^3$};
            \node (11_9) [left=.6cm of 11_10] {$v_{11}^2$};
            \node (11_7) [left=.7cm of 11_9] {$v_{11}^1$};
            \node (12_11) [below=.5cm of 11_10] {$v_{12}^3$};
            \node (12_10) [left=1.8cm of 12_11] {$v_{12}^2$};
            \node (12_9) [left=1.7cm of 12_10] {$v_{12}^1$};
            \node (13_12) [below=.5cm of 12_11] {$v_{13}^3$};
            \node (f) [below left=.5cm and 2.cm of 13_12] {$f$};
            
            \draw[->] (0_0) -- (1_0);
            \draw[->] (0_0) -- (3_1);
            \draw[->] (0_0) -- (5_2);
            \draw[->] (1_0) -- (2_1);
            \draw[->] (2_1) -- (4_2);
            \draw[->] (3_1) -- (4_3);
            \draw[->] (3_2) -- (7_3);
            \draw[->, dashed] (3_1) --  (3_2);
            \draw[->] (4_2) -- (7_5);
            \draw[->] (4_3) -- (9_5);
            \draw[->, dashed] (4_2) --  (4_3);
            \draw[->] (5_2) -- (7_5);
            \draw[->] (5_3) -- (9_5);
            \draw[->, dashed] (5_2) --  (5_3);
            \draw[->] (5_4) -- (6_5);
            \draw[->, dashed] (5_3) --  (5_4);
            \draw[->] (6_5) -- (11_7);
            \draw[->] (7_3) -- (9_7);
            \draw[->] (7_5) -- (11_7);
            \draw[->, dashed] (7_3) --  (7_5);
            \draw[->] (7_6) -- (8_7);
            \draw[->, dashed] (7_5) --  (7_6);
            \draw[->] (8_7) -- (12_9);
            \draw[->] (9_5) -- (11_9);
            \draw[->] (9_7) -- (12_9);
            \draw[->, dashed] (9_5) --  (9_7);
            \draw[->] (9_8) -- (10_9);
            \draw[->, dashed] (9_7) --  (9_8);
            \draw[->] (10_9) -- (f);
            \draw[->] (11_7) -- (12_11);
            \draw[->] (11_9) -- (f);
            \draw[->, dashed] (11_7) --  (11_9);
            \draw[->] (11_10) -- (f);
            \draw[->, dashed] (11_9) --  (11_10);
            \draw[->] (12_9) -- (f);
            \draw[->] (12_10) -- (f);
            \draw[->, dashed] (12_9) --  (12_10);
            \draw[->] (12_11) -- (13_12);
            \draw[->, dashed] (12_10) --  (12_11);
            \draw[->] (13_12) -- (f);
        \end{tikzpicture}}
        \\(a) MAS-Skeleton
    \end{minipage}\hfill
    \begin{minipage}[b]{.3\textwidth}
    \centering
    \scalebox{.9}{
        \begin{tikzpicture}
            \node (start) {$s$};
            \node (3) [below right=of start] {$3$};
            \node (5) [below left=of start] {$5$};
            \node (7) [below=of 3] {$7$};
            \node (6) [below=of 5] {$6$};
            \node (9) [right=of 6, left=of 7] {$9$};
            \node (11) [below= of 7] {$11$};
            \node (10) [below= of 6] {$10$};
            \node (12) [below=of 9] {$12$};
            \node (n1) [below right=1cm and .3cm of 12] {$13+1$};
            \node (n3) [below left=1cm and .3cm of 12] {$13+3$};
            \node (f) [below left=1cm and .3cm of n1] {$f$};

            \draw[->] (start) -- (5);
            \draw[->, dashed] (5) -- (3);
            \draw[->] (3) -- (7);
            \draw[->, dashed] (9) -- (7);
            \draw[->, dashed] (6) -- (9);
            \draw[->] (5) -- (6);
            \draw[->, dashed] (10) -- (12);
            \draw[->, dashed] (12) -- (11);
            \draw[->] (7) -- (11);
            \draw[->] (6) -- (11);
            \draw[->] (9) -- (10);
            \draw[->, dashed] (n3) -- (n1);
            \draw[->] (10) -- (n1); 
            \draw[->] (11) -- (n3);
            \draw[->] (12) -- (n3); 
            \draw[->] (n1) -- (f);
            \draw[->] (n3) -- (f);
        \end{tikzpicture}}
        \\(b) SAS-Skeleton
    \end{minipage}
    \caption{MAS-Skeleton (a) and SAS-Skeleton (b) of the word $11213\cdot3221\cdot132\cdot2$, of length $13$ (the $\cdot$ symbols separate the arches). For better readability, the MAS-Skeleton (a) does not show nodes that can not be reached from $s$.}\label{fig:massas1}
\end{figure}

\begin{figure}[tbh]
    \centering
    \begin{minipage}[b]{.6\linewidth}
        \centering
        \scalebox{.9}{\begin{tikzpicture}
            \node (0_0) {$s$};
            \node (1_0) [below left=.7cm and 1.2cm of 0_0] {$v_1^1$};
            \node (2_0) [below right=.5cm and 0cm of 1_0] {$v_2^1$};
            \node (2_1) [right=2.cm of 2_0] {$v_2^2$};
            \node (3_0) [below right=.5cm and 3cm of 2_0] {\textcolor{white}{$v_3^1$}};
            \node (3_2) [below right=.5cm and 1cm of 2_1] {$v_3^2$};
            \node (4_0) [below left=.5cm and 2.4cm of 3_0] {$v_4^1$};
            \node (4_2) [right=.6cm of 4_0] {$v_4^2$};
            \node (4_3) [right=.6cm of 4_2] {$v_4^3$};
            \node (5_0) [below=.5cm of 4_0] {\textcolor{white}{$v_5^1$}};
            \node (5_2) [right=1cm of 5_0] {\textcolor{white}{$v_5^2$}};
            \node (5_4) [below=.5cm of 4_3] {$v_5^3$};
            \node (6_2) [below left=.5cm and 0.5cm of 5_0] {$v_6^1$};
            \node (6_4) [right=.5cm of 6_2] {$v_6^2$};
            \node (6_5) [right=2.5cm of 6_4] {$v_6^3$};
            \node (7_2) [below=.5cm of 6_2] {\textcolor{white}{$v_7^1$}};
            \node (7_5) [right=1.5cm of 7_2] {$v_7^2$};
            \node (7_6) [right=1.9cm of 7_5] {$v_7^3$};
            \node (f) [below right=.5cm and .3cm of 7_5] {$f$};
            \draw[->] (0_0) --  (1_0);
            \draw[->] (0_0) --  (2_0);
            \draw[->] (0_0) --  (4_0);
            \draw[->] (1_0) --  (6_2);
            \draw[->] (2_0) --  (6_2);
            \draw[->] (2_1) --  (3_2);
            \draw[->, dashed] (2_0) --  (2_1);
            \draw[->] (3_2) --  (f);
            \draw[->] (4_0) --  (6_4);
            \draw[->] (4_2) --  (f);
            \draw[->, dashed] (4_0) --  (4_2);
            \draw[->] (4_3) --  (5_4);
            \draw[->, dashed] (4_2) --  (4_3);
            \draw[->] (5_4) --  (7_5);
            \draw[->] (6_2) --  (f);
            \draw[->] (6_4) --  (7_6);
            \draw[->, dashed] (6_2) --  (6_4);
            \draw[->] (6_5) --  (f);
            \draw[->, dashed] (6_4) --  (6_5);
            \draw[->] (7_5) --  (f);
            \draw[->] (7_6) --  (f);
            \draw[->, dashed] (7_5) --  (7_6);
        \end{tikzpicture}}
        \\(a) MAS-Skeleton
    \end{minipage}
    \hfill
    \begin{minipage}[b]{.35\linewidth}
        \centering 
        \scalebox{.9}{\begin{tikzpicture}
            \node (s) {$s$};
            \node (4) [below=of s] {$4$};
            \node (val2) [below=of 4] {$7+2$};
            \node (f) [below=of val2] {$f$};
            \draw[->] (s) -- (4);
            \draw[->] (4) -- (val2);
            \draw[->] (val2) -- (f);
        \end{tikzpicture}}
        \\(b) SAS-Skeleton
    \end{minipage}
    \caption{MAS-Skeleton (a) and SAS-Skeleton (b) of the word $1223\cdot 313$, of length $7$ (the $\cdot$ symbols separate the arches). Nodes, that can not be reached from $s$ in (a) have been removed for better readability. There is exactly one SAS.} \label{fig:massas2}
\end{figure}
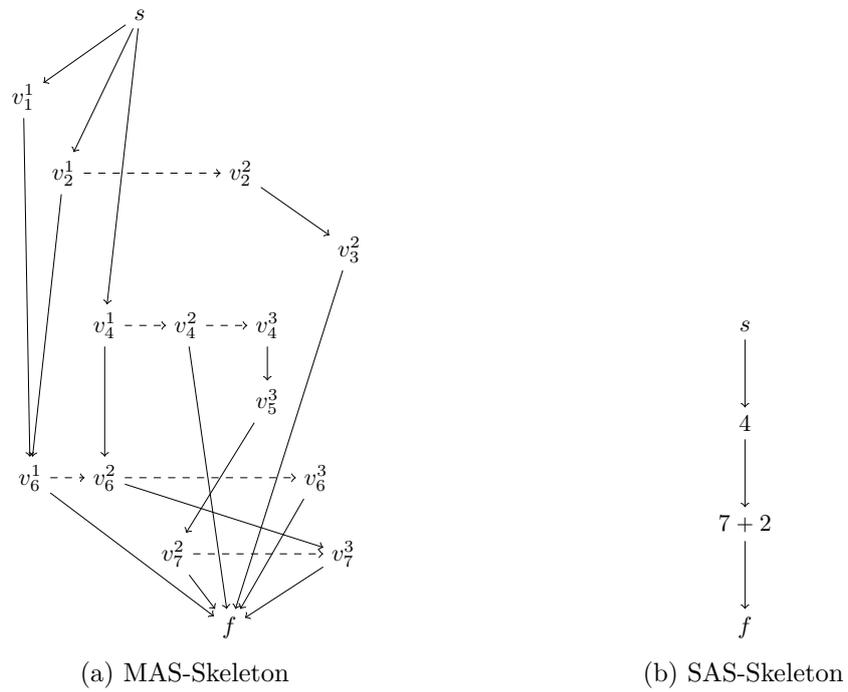

\end{document}